\documentclass{article}

\usepackage[final]{neurips_2024}

\usepackage[utf8]{inputenc} 
\usepackage[T1]{fontenc}    
\usepackage{hyperref}       
\usepackage{url}            
\usepackage{booktabs}       
\usepackage{amsfonts}       
\usepackage{nicefrac}       
\usepackage{microtype}      
\usepackage{xcolor}         

\usepackage{enumitem}
\usepackage{wrapfig}
\usepackage{amsmath}
\usepackage{amssymb}
\usepackage{mathtools}
\usepackage{amsthm}
\usepackage{bbm, dsfont}
\usepackage{nicefrac}
\allowdisplaybreaks

\usepackage{tikz}
\usepackage{standalone} 
\usepackage{tikzscale} 

\usepackage{algorithm}
\usepackage[noend]{algpseudocode}

\theoremstyle{plain}
\usepackage{thm-restate}
\newtheorem{theorem}{Theorem}[section]

\theoremstyle{remark}

\makeatletter
\newcommand{\commentsymbol}{\it\color{gray}$\triangleright$~}
\newcommand{\LComment}[1]{{\commentsymbol{#1}}}
\makeatother

\DeclareMathOperator*{\argmax}{arg\,max}
\newcommand{\preg}{\mathcal{P}}

\usepackage{pifont}

\usepackage{comment}

\newcommand{\nb}[3]{{\colorbox{#2}{\bfseries\sffamily\scriptsize\textcolor{white}{#1}}}{\textcolor{#2}{\sf\small\textit{#3}}}}

\newcommand{\albo}[1]{\nb{Albo}{cyan}{#1}}
\newcommand{\bollo}[1]{\nb{Bollo}{red}{#1}}

\title{The Sample Complexity of Stackelberg Games}

\author{
	Francesco Bacchiocchi,
	Matteo Bollini,
	Matteo Castiglioni,
	Alberto Marchesi \&
	Nicola Gatti  \,\\
	Politecnico di Milano\\
	{\textcolor{black}{\texttt{\{name.surname\}@polimi.it}}}
}

\begin{document}

\maketitle

\begin{abstract}
	\emph{Stackelberg games} (SGs) constitute the most fundamental and acclaimed models of strategic interactions involving some form of commitment.
	Moreover, they form the basis of more elaborate models of this kind, such as, \emph{e.g.}, Bayesian persuasion and principal-agent problems.
	Addressing \emph{learning} tasks in SGs and related models is crucial to operationalize them in practice, where model parameters are usually {unknown}.
	In this paper, we revise the \emph{sample complexity} of learning an optimal strategy to commit to in SGs.
	We provide a novel algorithm that (i) does \emph{not} require any of the limiting assumptions made by state-of-the-art approaches and (ii) deals with a trade-off between sample complexity and termination probability arising when leader's strategies representation has finite precision.
	Such a trade-off has been completely neglected by existing algorithms and, if \emph{not} properly managed, it may result in them using exponentially-many samples.
	%
	%
	%
	%
	Our algorithm requires novel techniques, which also pave the way to addressing learning problems in other models with commitment ubiquitous in the real world.
\end{abstract}

\section{Introduction}


Asymmetries are ubiquitous in strategic interactions that involve multiple agents.
The most fundamental and acclaimed models of asymmetric interactions are \emph{Stackelberg games} (SGs)~\citep{von1934marktform,Conitzer2006,von2010leadership}.
In an SG, a \emph{leader} has the ability to publicly \emph{commit to a strategy} beforehand, while a \emph{follower} reacts by best responding to it.
Such a simple idea of commitment is at the core of several other (more elaborate) models of asymmetric strategic interactions, such as Bayesian persuasion~\citep{kamenica2011bayesian}, principal-agent problems~\citep{myerson1982optimal}, and mechanism design~\citep{myerson1989mechanism}.

%

Recently, models of strategic interactions that involve some form of commitment have received a growing attention.
Addressing \emph{learning} tasks related to such models is of paramount importance to operationalize them in practice, where model parameters are usually \emph{unknown}.
Several works pursued this goal in various settings, ranging from SGs~\citep{Letchford2009,Blum2014learning,Balcan2015Commit,Blum2019Computing,Peng2019,Fiez2020,bai2021sample,lauffer2022noregret} to Bayesian persuasion~\citep{castiglioni2020online,DBLP:conf/sigecom/ZuIX21}
and principal-agent problems~\citep{ho2015adaptive,cohen2022learning,zhu2022online,bacchiocchi2023learning}.

In this paper, we revise the \emph{sample complexity} of learning an optimal strategy to commit to in~SGs.
\citet{Letchford2009} first addressed this learning problem, by providing an algorithm that works by ``sampling'' suitably-selected leader's strategies to get information about follower's best responses.
The main drawback of such an algorithm is that, in the worst case, it may require a number of samples growing exponentially in the number of leader's actions $m$ and in the representation precision $L$ (expressed in terms of number of bits) of players' payoffs.
\citet{Peng2019} later built on top of the results by~\citet{Letchford2009} to design an algorithm requiring a number of samples growing polynomially in $L$ and exponentially in either the number of leader's actions $m$ or $n$.
Moreover, \citet{Peng2019} provide a lower bound showing that such a sample complexity result is tight.

The algorithm by~\citet{Peng2019} relies on a number of rather stringent assumptions that severely limit its applicability in practice.
Specifically, it assumes to have control over the action played by the follower when multiple best responses are available to them, and, additionally, that follower's payoffs satisfy some suitable non-degeneracy conditions.
Moreover, even if all its assumptions are met, the algorithm by~\citet{Peng2019} may still fail in some SGs that we showcase in this paper.

Another issue of the algorithm by~\citet{Peng2019} is that its theoretical guarantees crucially rely on the assumption that leader's strategies can be selected uniformly at random \emph{without} taking into account their representation precision.
This may result in the algorithm requiring an exponential number of samples, as we show in this paper.
Properly accounting for the representation precision of leader's strategies requires managing a challenging trade-off between the number of samples used by the algorithm and the probability with which it terminates.
Understanding how to deal with such a trade-off is fundamental, \emph{not} only for the problem of learning optimal commitments in SGs, but also to build a solid basis to tackle related learning problems in other models involving commitment.

In this paper, we introduce a new algorithm to learn an optimal strategy to commit to in SGs.
Our algorithm requires a number of samples that scales polynomially in $L$ and exponentially in either $m$ or $n$---this is tight due to the lower bound by~\citet{Peng2019}---, it does \emph{not} require any of the limiting assumptions made by~\citet{Peng2019}, and it circumvents all the issues of their algorithmic approach by properly managing the representation precision of leader's strategies.

\section{Preliminaries}\label{sec:preliminaries}


A \emph{(normal-form) Stackelberg game} (SG) is defined as a tuple $G \coloneqq ({A}_\ell, {A}_f, u_\ell, u_f )$, where: ${A}_\ell \coloneqq \{a_i \}_{i=1}^m$ is a finite set of $m$ leader's actions, ${A}_f \coloneqq \{a_j \}_{j=1}^n$ is a finite set of $n$ follower's actions, while $u_\ell, u_f : {A}_\ell \times\mathcal{A}_f \to \mathbb{Q} \cap [0,1]$ are leader's and follower's utility functions, respectively.
Specifically, $u_\ell(a_i,a_j)$ and $u_f(a_i,a_j)$ are the payoffs obtained by the players when the leader plays $a_i \in A_\ell$ and the follower plays $a_j \in A_f$. 
%
%
A leader's \emph{mixed strategy} is a probability distribution $p \in \Delta_{m}$ over leader's actions, with $p_i$ denoting the probability of action $a_i \in A_\ell$.
The space of all leader's mixed strategies is the $(m-1)$-dimensional simplex, namely $\Delta_m \coloneqq \{ p \in \mathbb{R}_+^m \mid \sum_{a_i \in A_\ell} p_i = 1 \}$.

In an SG, the leader commits to a {mixed strategy} beforehand, and the follower decides how to play after observing it.
Given a leader's commitment $p \in \Delta_{m}$, we can assume w.l.o.g.~that the follower plays an action deterministically.
%
In particular, the follower plays a \emph{best response}, which is an action maximizing their expected utility given $p$.
%
%
%
Formally, the set of follower's best responses is
\[
		A_f(p) \coloneqq \argmax_{a_j \in  {A}_f}	 \sum_{a_i \in A_\ell} p_i u_f(a_i,a_j) .
\]
%
%
As customary in the literature~\citep{Conitzer2006}, we assume that the follower breaks ties in favor of the leader when having multiple best responses available, choosing an action maximizing leader's expected utility.
Formally, after observing $p \in \Delta_{m}$, the follower plays an action $a_f^\star(p) \in {A}_f$ such that $a_f^\star(p) \in \argmax_{a_j \in A_f(p)} \sum_{a_i \in A_\ell} p_i u_\ell(a_i,a_j).$

In an SG, the goal of the leader is to find an \emph{optimal strategy to commit to}, which is one maximizing their expected utility given that the follower always reacts with a best response.
%
Formally, the leader faces the following bi-level optimization problem: $\max_{p \in \Delta_m} u_\ell(p)$, in which, for ease of notation, we let $u_\ell(p) \coloneqq \sum_{a_i \in A_\ell}p_i u_\ell(a_i,a_f^\star(p))$ be leader's expected utility by committing to $p \in \Delta_{m}$.
%

Follower's best responses induce a cover of $\Delta_{m}$, composed by a family of $n$ best-response regions $\mathcal{P}_j$ defined as follows.
Given any pair of follower's actions $a_j, a_k \in A_f : a_j \neq a_k$, we let $\mathcal{H}_{j k} \subseteq \mathbb{R}^m$ be the halfspace in which $a_j$ is (weakly) better than $a_k$ in terms of follower's utility, where:
\[
	\mathcal{H}_{j k} \hspace{-0.5mm} \coloneqq \hspace{-0.5mm} \left\{  p \in \mathbb{R}^{m} \hspace{-0.5mm} \mid \hspace{-0.5mm} \sum_{a_i \in A_\ell} \hspace{-0.5mm} p_i \big( u_f(a_i,a_j) \hspace{-0.5mm} - \hspace{-0.5mm} u_f(a_i, a_k) \big) \hspace{-0.5mm} \geq \hspace{-0.5mm} 0  \right\} \hspace{-0.5mm} .
\]
Moreover, we denote by $H_{j k} \coloneqq \partial \mathcal{H}_{jk}$ the hyperplane constituting the boundary of the halfspace $\mathcal{H}_{jk}$, which we call the \emph{separating hyperplane} between $a_j$ and $a_k$.\footnote{We let $\partial \mathcal{H}$ be the boundary hyperplane of halfspace $\mathcal{H} \subseteq \mathbb{R}^m$.
Notice that $H_{jk}$ and $H_{kj}$ actually refer to the same hyperplane. In this paper, we use both names for ease of presentation.}
Then, for every follower's action $a_j \in A_f$, we define its \emph{best-response} region $\mathcal{P}_j \subseteq \Delta_{m}$ as the subspace of leader's strategies in which action $a_j$ is a best response.
The set $\mathcal{P}_j$ is defined as the intersection of $\Delta_{m}$ with \emph{all} the halfspaces $\mathcal{H}_{jk}$ in which $a_j$ is better than another action $a_k \in A_f$. Formally:
\[
	\mathcal{P}_j \coloneqq \Delta_{m} \cap \Bigg(  \bigcap_{a_k \in A_f: a_k \neq a_j} \mathcal{H}_{jk} \Bigg).
\]
The family of all sets $\mathcal{P}_j$ constitutes a cover of $\Delta_{m}$, since $\Delta_{m} = \bigcup_{a_j \in A_f} \mathcal{P}_j$.
%
Notice that $\mathcal{P}_j$ is a polytope whose vertices are obtained by intersecting $\Delta_m$ with $m-1$ linearly-independent hyperplanes, selected among separating hyperplanes $H_{jk}$ and \emph{boundary hyperplanes} of~$\Delta_{m}$.
The latter are defined as $H_i \coloneqq \{ p \in \mathbb{R}^m \mid p_i = 0 \}$ for every $a_i \in A_\ell$.
Since there are at most $n-1$ separating hyperplanes and $m$ boundary ones, the vertices of $\mathcal{P}_j$ are at most $\binom{n+m}{m}$.
In the following, we let $V(\mathcal{P}_j) \subseteq \Delta_{m}$ be the set of all the vertices of $\mathcal{P}_j$, while we denote by $\text{vol}(\mathcal{P}_j)$ its volume relative to $\Delta_{m}$.
%
Moreover, we let $\text{int}(\mathcal{P}_j)$ be the interior of $\mathcal{P}_j$ relative to $\Delta_{m}$.\footnote{We denote by $\text{vol}_{d}(\mathcal{P})$ the Lebesgue measure in $d$ dimensions of a polytope $\mathcal{P} \subseteq \mathbb{R}^D$. For ease of notation, whenever $D = m$ and $d = m-1$, we simply write $\text{vol}(\mathcal{P})$. Moreover, we let $\text{int}(\mathcal{P})$ be the interior of $\mathcal{P}$ relative to a subspace that fully contains $\mathcal{P}$ and has minimum dimension. In the case of a best-response region $\mathcal{P}_j$, the $(m-1)$-dimensional simplex is one of such subspaces.}
%
%
%
%
%
Once all the best-response regions are available, the optimization problem faced by the leader can be formulated as $\max_{a_j \in A_f} \max_{p \in \mathcal{P}_j} u_\ell(p)$, where the inner $\max$ can be solved efficiently by means of an LP~\citep{Conitzer2006}.

\paragraph{Learning in Stackelberg games}
We study SGs in which the leader does \emph{not} know anything about follower's payoffs and they have to \emph{learn} an optimal strategy to commit to.
%
%
The leader can only interact with the follower by committing to a strategy $p \in \Delta_{m}$ and observing the best response $a_f^\star(p)$ played by the latter.
%
%
We assume that the leader interacts with the follower by calling a function \texttt{Oracle}$(p)$, which takes $p \in \Delta_{m}$ and returns $a^\star_f(p)$.
%
Since we are concerned with the \emph{sample complexity} of learning an optimal strategy to commit to in SGs, our goal is to design algorithms that the leader can employ to learn such a strategy by using the minimum possible number of \emph{samples}, \emph{i.e.}, calls to \texttt{Oracle}$(p)$.
Ideally, we would like algorithms that require a number of samples scaling polynomially in the parameters that define the size of SGs, \emph{i.e.}, the number of leader's actions $m$, that of follower's ones $n$, and the number of bits encoding players' payoffs.
However, a lower bound by~\citet{Peng2019} shows that an exponential dependence in either $m$ or $n$ is unavoidable.
Thus, as done by~\citet{Peng2019}, we pursue the goal of learning an optimal strategy to commit to by using a number of samples that is polynomial when either $m$ or $n$ is fixed.

\paragraph{On the representation of numbers}
Throughout the paper, we assume that all the numbers manipulated by our algorithms are rational.
%
%
%
We assume that rational numbers are represented as fractions, by specifying two integers which encode their numerator and denominator~\citep{Schrijver1986}.
%
Given a rational number $q  \in \mathbb{Q}$ represented as a fraction $\nicefrac{b}{c}$ with $b,c \in \mathbb{Z}$, we denote the number of bits that $q$ occupies in memory, called \emph{bit-complexity}, as $B_{\nicefrac{b}{c}} := B_b + B_c$, where $B_b$ ($B_c$) is the number of bits required to represent the numerator (denominator).
%
%
For ease of presentation and with an abuse of terminology, given a vector in $\mathbb{Q}^D$ of $D$ rational numbers represented as fractions, we let its bit-complexity be the maximum bit-complexity among its entries.

\section{Sate-of-the-art of learning in Stackelberg games}\label{sec:sota}

We start by presenting the algorithm by~\citet{Peng2019}, which is the state-of-the-art approach to learn an optimal strategy to commit to in SGs.
We describe in detail the limiting assumptions made by such an algorithm, and we showcase an SG in which the algorithm fails even if all its assumptions are met.
Moreover, we discuss all the issues resulting from the fact that~\citet{Peng2019} do \emph{not} account for the bit-complexity of leader's strategies in their analysis.

\subsection{High-level description of the algorithm by~\citet{Peng2019}}

The core idea underpinning the algorithm by~\citet{Peng2019} is to iteratively discover separating hyperplanes $H_{jk}$ to identify all the best-response regions $\mathcal{P}_j$.
The algorithm keeps track of some overestimates---called \emph{upper bounds}---of the best-response regions, built by using the separating hyperplanes discovered so far.
All the upper bounds are initialized to $\Delta_{m}$, and, whenever a separating hyperplane $H_{jk}$ is discovered, the upper bounds of both $\mathcal{P}_j$ and $\mathcal{P}_k$ are updated accordingly.
Moreover, the algorithm also keeps track of some underestimates---called \emph{lower bounds}---of the best-response regions, with the lower bound of $\mathcal{P}_j$ containing all the points that have already been discovered to belong to $\mathcal{P}_j$.
After finding a new separating hyperplane, the algorithm checks each vertex $p \in \Delta_{m}$ of the (updated) upper bounds, adding it to the lower bound of the best-response region $\mathcal{P}_j$ with $a_j = a^\star_f(p)$ (by taking its convex hull with the lower bound).

To find a separating hyperplane, the algorithm by~\citet{Peng2019} performs a \emph{binary search} over suitably-defined line segments in $\Delta_{m}$, by using a procedure introduced by~\citet{Letchford2009}.
As a first step, the procedure does a binary search on a line segment connecting two points $p^1, p^2 \in \Delta_{m}$ such that $a^\star_f(p^1) \neq a^\star_f(p^2)$.
In particular, the procedure identifies two actions $a_j, a_k \in A_f$ whose upper bounds have an intersection with non-zero volume relative to $\Delta_{m}$, and it randomly selects $p^1$ from the interior of such an intersection.
Moreover, $p^2$ is any point in the lower bound of either $a_j$ and $a_k$, chosen so that $a^\star_f(p^1) \neq a^\star_f(p^2)$.
The binary search recursively halves the line segment connecting $p^1$ and $p^2$ until it finds a point $p^\circ \in \Delta_{m}$ on some separating hyperplane (notice that this could be potentially different from $H_{jk}$).
Then, the procedure draws a ``small'' $(m-1)$-dimensional simplex centered at $p^\circ$ uniformly at random, it suitably chooses $m-1$ distinct pairs of vertices of such a simplex, and it performs binary search on the line segments defined by such pairs, so as to find additional $m-1$ points on the separating hyperplane and identify it.

\subsection{Assumptions of the algorithm by~\citet{Peng2019}}

The algorithm by~\citet{Peng2019} crucially relies on the following assumptions:
\begin{itemize}[noitemsep,nolistsep]
	\item[(a)] each best-response region $\mathcal{P}_j$ is either empty or it has volume $\text{vol}(\mathcal{P}_j)$ greater than or equal to $2^{-nL}$, where $L$ is the bit-complexity of follower's payoffs;
	\item[(b)] when indifferent, the follower breaks ties as needed by the algorithm to correctly terminate;
	\item[(c)] there are no $m+1$ separating/boundary hyperplanes that intersect in one point; and 
	\item[(d)] no separating hyperplanes coincide.
\end{itemize}
%
%
Notice that assumptions (c)~and~(d) considerably limit the set of SGs in which the algorithm can be applied, and assumption~(a) makes such limits even stronger by ruling out cases where a follower's action is a best response only for a ``small'' subset of leader's strategies. 
%
%
Moreover, assumption (b) is a rather unreasonable requirement in practice, since it amounts to assuming that the algorithm has some form of control on which best responses are played by the follower.

The algorithm by~\citet{Peng2019} needs assumptions (a),~(c),~and~(d) in order to meet the requirements of the procedure introduced by~\citet{Letchford2009}, while assumption~(b) is needed to properly build the lower bounds of best-response regions by vertex enumeration.
As we show in the rest of this paper, our algorithm employs novel techniques that allow to drop assumptions (a)--(d).
\begin{wrapfigure}[18]{R}{0.36\textwidth}
	\vspace{-.7cm}
	\begin{minipage}{0.43\textwidth}
		\begin{figure}[H]
			\centering
			\includegraphics[width=0.75\linewidth]{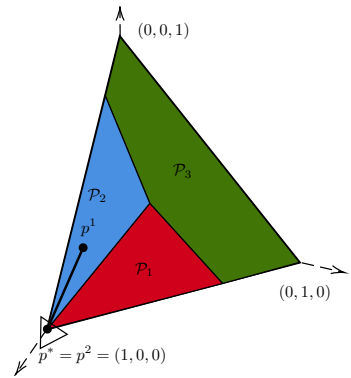}
			\caption{Algorithm by~\citet{Peng2019} fails even if all its assumptions are met.}
			\label{fig:simplex_1}
		\end{figure}
	\end{minipage}
\end{wrapfigure}


\paragraph{Minor issues of the algorithm by~\citet{Peng2019}}
We provide an instance of SG in which, even if assumptions (a)--(d) are met, the algorithm by~\citet{Peng2019} fails to find an optimal commitment.
The failure is depicted in Figure~\ref{fig:simplex_1}.
There, the algorithm fails since binary search finds a point $p^\circ$ on a boundary hyperplane.
This clearly makes the construction of a ``small'' $(m-1)$-dimensional simplex centered at $p^\circ$ impossible.
Intuitively, this happens when the procedure by~\citet{Letchford2009} randomly draws a leader's strategy from the intersection of the upper bounds of $\mathcal{P}_1$ and $\mathcal{P}_2$ obtaining a point $p^1$ in $\mathcal{P}_2$.
Then, since at that point of the algorithm execution the lower bound of $\mathcal{P}_1$ only contains a point $p^2$ on a facet of $\mathcal{P}_1$, binary search is performed on a segment fully contained in the best-response region $\mathcal{P}_2$.
Thus, binary search inevitably ends at $p^\circ = p^2$, resulting in a failure.
Additional details on this example are provided in Appendix~\ref{sec:problem_prev_works}, where we also showcase another minor issue of the procedure by~\citet{Letchford2009}, inherited by~\citet{Peng2019}.

\subsection{The algorithm by~\citet{Peng2019} may require exponentially-many samples}

The analysis of the algorithm by~\citet{Peng2019} relies on the assumption that leader's strategies can be selected uniformly at random \emph{without} taking into account their bit-complexity.
This allows \citet{Peng2019} to claim that the binary searches performed by their algorithm only require $\mathcal{O}(L)$ samples, where $L$ is the bit-complexity of follower's payoffs.
However, as we show in the rest of this paper, the number of samples required by those binary searches also depends on the bit-complexity of leader's strategies employed by the algorithm, and, if this is \emph{not} properly controlled, the algorithm by~\citet{Peng2019} may require an exponential number of samples.
\begin{wrapfigure}[17]{R}{0.43\textwidth}
	\vspace{-.8cm}
	\begin{minipage}{0.43\textwidth}
		\begin{figure}[H]
			\centering
			\includegraphics[width=0.75\linewidth]{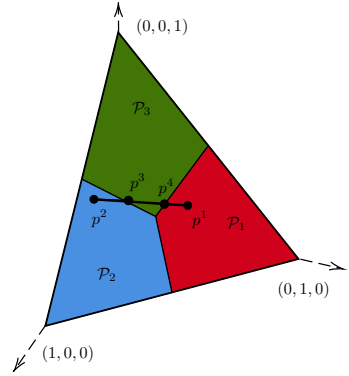}
			\caption{Algorithm by~\citet{Peng2019} requires exponentially-many samples.}
			\label{fig:simplex_2}
		\end{figure}
	\end{minipage}
\end{wrapfigure}

Figure~\ref{fig:simplex_2} shows an example of algorithm execution suffering the issue described above.
The complete example is provided in Appendix~\ref{sec:problem_prev_works}.
Intuitively, in Figure~\ref{fig:simplex_2}, a binary search is performed on the line segment connecting points $p^1$ and $p^2$, which are leader's strategies parametrized by some $\epsilon > 0$.
The line segment intersects two separating hyperplanes, at points $p^3$ and $p^4$, and the distance between $p^3$ and $p^1$ can be made arbitrarily small by lowering $\epsilon$.
%
%
To identify the separating hyperplane passing through $p^4$, the binary search procedure must go on until a point on the segment connecting $p^3$ and $p^1$ is found.
This requires a number of steps logarithmic in the distance between $p^1$ and $p^2$ divided by the one between $p^1$ and $p^3$.
Thus, binary search requires $\mathcal{O}(\log (\nicefrac{1}{\epsilon}))$ samples, and, whenever $\epsilon = \mathcal{O}(\text{exp}(-2^L))$, it requires $\mathcal{O}(2^L)$ samples.
Intuitively, the issue arises from the fact that the bit-complexity of the point $p^1$ is exponential in $L$.
This may happen any time $p^1$ is randomly selected \emph{without} properly controlling the bit-complexity of leader's strategies.
As we show in the rest of this paper, in order to avoid the issue depicted in Figure~\ref{fig:simplex_2}, it is necessary that leader's strategies have a properly-controlled (finite) bit-complexity.
This introduces a trade-off between the number of samples required by the algorithm and its termination probability, which must be suitably managed.

\section{Learning an optimal commitment is SGs}


At a high level, our algorithm works by \emph{closing} follower's actions one after the other, with an action being considered closed when a set of separating hyperplanes identifying its best-response region has been found.
The algorithm stops when the union of the identified best-response regions coincides with $\Delta_{m}$.
In this way, it builds all the best-response regions $\mathcal{P}_j$ with $\text{vol}(\mathcal{P}_j) > 0$, which we show to be sufficient to find an optimal strategy to commit to.
To close an action $a_j \in {A}_f$, our algorithm builds an upper bound $\mathcal{U}_j \subseteq \Delta_{m}$ of the best-response region $\mathcal{P}_j$, using the separating hyperplanes $H_{jk}$ computed so far.
To find a separating hyperplane, the algorithm adopts a binary-search-based procedure similar in nature to the one by~\citet{Letchford2009}.
However, our procedure has some crucial differences that allow to avoid the issue depicted in Figure~\ref{fig:simplex_2}.
In particular, it ensures that the point randomly drawn from the intersection of two upper bounds has a properly-controlled bit-complexity.
This is accomplished by using a suitable sampling technique, which manages the trade-off between number of samples and termination probability. 
%
%
%
Moreover, to understand whether action $a_j$ has been closed or \emph{not}, the algorithm checks all the vertices of $\mathcal{U}_j$ to understand if $a_j$ is a best response in each of them, with an idea similar to that in~\citep{Peng2019}.
However, to relax the stringent assumption~(b) made by~\citet{Peng2019}, our algorithm checks a vertex by querying a nearby point that is in the interior of the upper bound.
Such a point is obtained by moving from the vertex towards a known point in the interior of the best-response region.
Crucially, provided that the queried point is sufficiently close to the vertex, if a follower's action is a best response in the queried point, then it is also a best response in the vertex, and \emph{viceversa}.
Moreover, avoiding querying vertexes also ensures that the failure of the algorithm by~\citet{Peng2019} depicted in Figure~\ref{fig:simplex_1} does \emph{not} occur in our algorithm, since it guarantees that a vertex of the currently-considered upper bound is used in a binary search only when it is \emph{not} on the boundary of its associated best-response region.
%

Next, we describe all the components of our algorithm.
Section~\ref{sec:main_algo} introduces the main procedure executed by the algorithm, called \texttt{Learn-Optimal-Commitment}.
Sections~\ref{sec:find_hyperplane}~and~\ref{sec:binary_search} introduce \texttt{Find-Hyperplane} and \texttt{Binary-Search}, respectively, which are two procedures working in tandem to find separating hyperplanes.
Finally, Section~\ref{sec:sample point} describes \texttt{Sample-Int}, which is a procedure called any time the algorithm has to randomly sample a point in the interior of some polytope.


\subsection{\texttt{Learn-Optimal-Commitment}}\label{sec:main_algo}


The pseudocode of \texttt{Learn-Optimal-Commitment} is provided in Algorithm~\ref{alg:learning_commitment}.
Notice that the algorithm takes as input a parameter $\zeta \in (0,1)$, which is used to control the bit-complexity of leader's strategies selected at random, thus managing the trade-off between number of samples and probability of correctly terminating, as shown by our results in the following.\footnote{For ease of presentation, we assume that there are \emph{no} $a_j, a_k \in A_f$ such that $u_f(a_i,a_j) = u_f(a_i,a_k)$ for all $a_i \in A_\ell$. Indeed, if such two actions exist, their best-response regions $\mathcal{P}_j$ and $\mathcal{P}_k$ coincide, and, thus, in any $p \in \mathcal{P}_j \equiv \mathcal{P}_k$ the best response $a_f^\star(p)$ depends on leader's payoffs (according to tie breaking). As shown in Appendix~\ref{sec:app_extension}, our algorithm can be extended to also deal with such cases, and all our results continue to hold.}
\begin{wrapfigure}[24]{R}{0.56\textwidth}
	\vspace{-0.85cm}
	\begin{minipage}{0.56\textwidth}
		\begin{algorithm}[H]
			\caption{\texttt{Learn-Optimal-Commitment}}\label{alg:learning_commitment}
			\small
			\begin{algorithmic}[1]
				\Require Parameter $\zeta \in (0,1)$
				\State $B \gets$ Bit-complexity of \texttt{Sample-Int}\hfill
				\LComment{\textnormal{See Lemma~\ref{lem:sample_point_first}}}
				\State $\lambda \gets m 2^{-m(B+4L)-1}$  \hfill \LComment{\textnormal{See Lemma~\ref{lem:vertex_two}}}
				\State $\delta \gets \nicefrac{\zeta}{2(n^2 + nm)^2+n^3}$\label{line:delta}
				\State $\mathcal{C} \gets \varnothing$ \hfill \LComment{\textnormal{Set of \emph{closed} follower's actions}}
				\While{$\bigcup_{a_j \in \mathcal{C}} \mathcal{U}_{j} \ne \Delta_{m}$ }\label{line:partition_loop1}
				\State $p^{\text{int}} \gets $ Sample a point from $\text{int} \big(\Delta_m \setminus \bigcup_{a_k \in \mathcal{C}} \mathcal{U}_{k} \big) $ \label{line:sample_pint}
				\State $a_j \gets \texttt{Oracle}(p^\text{int})$
				\State $\mathcal{U}_{j} \gets \Delta_m$  \hfill \LComment{\textnormal{Initialize upper bound of $\mathcal{P}_j$}} 
				\State $\mathcal{V} \gets V(\mathcal{U}_{j})$ \hfill \LComment{\textnormal{Set of unchecked vertices of $\mathcal{U}_j$}}\label{line:set_vertexes}
				\While{$\mathcal{V} \neq \varnothing$}\label{line:partition_loop3}
				\State $v \gets $ Take any vertex in $\mathcal{V}$
				\State $p \gets \lambda p^\text{int} + (1-\lambda) v$ \label{line:convex_comb}
				\State $a \leftarrow \texttt{Oracle}(p)$\label{line:d_j}
				\If{$a \neq a_j$}
				\State $ H_{jk} \gets\texttt{Find-Hyperplane} (a_j,\mathcal{U}_{j}, p^\text{int}, v, \delta)$
				\State $\mathcal{U}_{j} \leftarrow \mathcal{U}_{j} \cap \mathcal{H}_{jk}$  \hfill \LComment{\textnormal{Update upper bound}}
				\State $\mathcal{V} \gets V(\mathcal{U}_j)$ \hfill \LComment{\textnormal{Update unchecked vertices}} \label{line:re_init_vertexes}
				\Else
				\State $\mathcal{V} \gets \mathcal{V} \setminus \{v\}$\hfill \LComment{\textnormal{Vertex $v$ has been checked}}
				\EndIf
				\EndWhile
				\State $\mathcal{C} \gets \mathcal{C}  \cup \{a_j\}$ \hfill \LComment{\textnormal{Action $a_j$ has been closed}}
				\EndWhile
				\State   $p^\star \gets \argmax_{p \in  \bigcup_{a_j \in \mathcal{C}} V(\mathcal{U}_{j} )}  u_\ell(p)$\label{line:optimal}
			\end{algorithmic}
		\end{algorithm}
	\end{minipage}
\end{wrapfigure}

During its execution, Algorithm~\ref{alg:learning_commitment} tracks already-closed follower's actions in a set $\mathcal{C}$.
If there are still actions that have to be closed, \emph{i.e.}, $\bigcup_{a_j \in \mathcal{C}} \mathcal{U}_j \neq \Delta_{m}$, the algorithm identifies one of them by randomly sampling a point from the interior of $\Delta_m \setminus \bigcup_{a_k \in \mathcal{C}} \mathcal{U}_{k} $, containing leader's strategies that have \emph{not} been covered yet by already-found best-response regions (Line~\ref{line:sample_pint}).\footnote{Notice that $\bigcup_{a_j \in \mathcal{C}} \mathcal{U}_j \neq \Delta_{m}$ may \emph{not} be convex in general. However, it can be expressed as the union of a finite number of polytopes. As we show in Appendix~\ref{sec:app_running_time}, this allows to suitably apply the \texttt{Sample-Int} procedure described in Section~\ref{sec:sample point} to get the needed point in polynomial time, when either $m$ or $n$ is fixed.}
With high probability, the sampling step provides a point $p^\text{int} \in \Delta_{m}$ in the interior of some best-response region $\mathcal{P}_j$ with $a_j = \texttt{Oracle}(p^\text{int})$ and $a_j \notin \mathcal{C}$.
Then, the algorithm focuses on closing action $a_j$.
First, it initializes the upper bound $\mathcal{U}_j$ of $\mathcal{P}_j$ to be $\Delta_{m}$.
Then, it iterates over the vertices $V(\mathcal{U}_j)$ of the upper bound $\mathcal{U}_j$, by employing a set $\mathcal{V}$ containing all the vertices that have \emph{not} been checked yet (Line~\ref{line:set_vertexes}).
The algorithm checks a vertex $v \in \mathcal{V}$ by querying a suitable convex combination $p$ of the interior point $p^\text{int}$ and $v$ (Line~\ref{line:convex_comb}), to get the best response $a = \texttt{Oracle}(p)$ played by the follower.
If the algorithm finds an action $a \neq a_j$, it means that a new separating hyperplane has to be discovered.
This is done by calling the \texttt{Find-Hyperplane} procedure, which takes as input $a_j$, the upper bound $\mathcal{U}_j$, the internal point $p^\text{int}$, and the vertex $v$, and it works as described in Section~\ref{sec:find_hyperplane}.
As shown later in this section, with high probability, the procedure returns a new separating hyperplane $H_{jk}$, for some action $a_k \in A_f$ possibly different from $a$.
%
Given $H_{jk}$, the algorithm updates the upper bound by intersecting it with the halfspace $\mathcal{H}_{jk}$ identified by $H_{jk}$.
This may change the vertices of $\mathcal{U}_j$, by either adding new ones or removing old ones.
Thus, the algorithm re-initializes $\mathcal{V}$ to contain all the vertices $V(\mathcal{U}_j)$ of the new upper bound (Line~\ref{line:re_init_vertexes}), and it starts checking vertices again.
%
When the algorithm has checked all the vertices in $\mathcal{V}$ obtaining best responses $a = a_j$, it means that all the separating hyperplanes defining the best-response region $\mathcal{P}_j$ have been identified.
Thus, it adds $a_j$ to the set $\mathcal{C}$ and goes on with a follower's action that still has to be closed (if any).

A crucial step of Algorithm~\ref{alg:learning_commitment} is to check a vertex $v \in V(\mathcal{U}_j)$ by querying a point nearby $v$ in the interior of $\mathcal{U}_j$.
This is obtained by moving towards the direction of the interior point $p^\text{int}$ of the best-response region $\mathcal{P}_j$.
This is crucial to drop assumption~(b), which is instead needed by the vertex enumeration procedure by~\citet{Peng2019}.
Indeed, thanks to such an assumption, when the algorithm by~\citet{Peng2019} queries a vertex of an upper bound in which there are multiple best responses available, it can choose \emph{any} of them as needed in order to complete vertex enumeration.
Without assumption~(b), $\texttt{Oracle}(v)$ for a vertex $v \in V(\mathcal{U}_j)$ may return a follower's action $a_k \neq a_j$ even though action $a_j$ is also a best response in $v$ (due to tie-breaking).
Our algorithm circumvents the issue by querying a nearby point in the interior of the upper bound.
The following lemma states that such a trick works provided that the point is sufficiently near the vertex. 
Intuitively, the lemma shows that, if $a_j$ is a best response in a leader's strategy sufficiently close to the vertex in the direction of the interior of $\mathcal{P}_j$, then it is also a best response in the vertex itself, and \emph{viceversa}.
%
%
Formally:
\begin{restatable}{lemma}{vertextwo}
	\label{lem:vertex_two}
	Given two points $p \in \textnormal{int}(\mathcal{P}_{j})$ with $a_j \in {A}_f$ and $p' \in \Delta_{m}$, each having bit-complexity bounded by $B$, let $\widetilde p\coloneqq \lambda p + (1- \lambda)p'$
	for some $\lambda\in (0,2^{-m(B+4L)-1}) $.
	Then:
	$
		\widetilde p \in \mathcal{P}_{j} \Leftrightarrow p' \in \mathcal{P}_{j}.
	$
	%
\end{restatable}

The two fundamental properties guaranteed to hold when Algorithm~\ref{alg:learning_commitment} terminates its execution are: (i) each upper bound $\mathcal{U}_j$ such that $a_j \in \mathcal{C}$ coincides with the best-response region $\mathcal{P}_j$, and (ii) all follower's actions $a_j \in {A}_f$ whose best-response regions have volume $\textnormal{vol}(\mathcal{P}_j) > 0$ have been closed by the algorithm.
As we show in the following, Algorithm~\ref{alg:learning_commitment} terminates with high probability  with properties (i)~and~(ii) satisfied.
The algorithm does \emph{not} terminate whenever either the sampling step in Line~\ref{line:sample_pint} does \emph{not} give a point in the interior of $\mathcal{P}_j$ or \texttt{Find-Hyperplane} is \emph{not} able to find a new separating hyperplane.
The probability of such events happening can be made arbitrarily low by increasing the bit-complexity of the points produced by \texttt{Sample-int}, as discussed in Section~\ref{sec:sample point}.

Next, we show that the properties of Algorithm~\ref{alg:learning_commitment} are sufficient to find an optimal strategy to commit to, by checking all the vertices of the upper bounds and taking the one providing the highest leader's expected utility (Line~\ref{line:optimal}).
This follows from a fundamental property of SGs, which holds since the follower breaks ties in leader's favor.
Intuitively, there is always an optimal strategy to commit that coincides with a vertex of a best-response region $\mathcal{P}_j$ with strictly-positive volume.
Formally:
%
%
\begin{restatable}{lemma}{VertexStackelberg}\label{lem:vertex_optimal}
	Given an SG, there exists an optimal strategy to commit to $p^\star \in \Delta_{m}$ and a follower's action $a_j \in {A}_f$ such that $\textnormal{vol}(\mathcal{P}_j) > 0$ and $p^\star = p$ for some vertex $p \in V(\mathcal{P}_j)$ of $\mathcal{P}_j$.
	%
\end{restatable}

Finally, by means of Lemma~\ref{lem:vertex_two}, Lemma~\ref{lem:vertex_optimal}, and all the results related to \texttt{Find-Hyperplane} and \texttt{Sample-Int}, presented in Sections~\ref{sec:find_hyperplane}~and~\ref{sec:sample point}, respectively, we can prove the following theorem, which also provides a bound on the number of samples required by Algorithm~\ref{alg:learning_commitment}.
\begin{restatable}{theorem}{maintheorem}\label{thm:main_thm}
	Given any $\zeta \in (0,1)$, with probability at least $1-\zeta$, Algorithm~\ref{alg:learning_commitment} terminates with $p^\star$ being an optimal strategy to commit to, by using $\widetilde{\mathcal{O}}\left(n^2\left(m^7L\log(\nicefrac{1}{\zeta})+ \binom{m+n}{m} \right)\right)$ samples. 
	%
	%
\end{restatable}
By Theorem~\ref{thm:main_thm}, the number of samples required by Algorithm~\ref{alg:learning_commitment} is polynomial when either the number of leader's actions $m$ or that of follower's actions $n$ is fixed.
%
%
%
Notice that Algorithm~\ref{alg:learning_commitment} requires an overall running time that is polynomial when either $m$ or $n$ is fixed (see Appendix~\ref{sec:app_running_time}).
%


\subsection{\texttt{Find-Hyperplane}}\label{sec:find_hyperplane}

The pseudocode of the \texttt{Find-Hyperplane} procedure is in Algorithm~\ref{alg:hyperplane}.
Algorithm~\ref{alg:learning_commitment} ensures that it receives as input an action $a_j \notin \mathcal{C}$ that has to be closed, the upper bound $\mathcal{U}_{j}$ of $\mathcal{P}_j$, an interior point $p^\text{int} \in \text{int}(\mathcal{P}_j)$, a vertex $v \in V(\mathcal{U}_j)$ with $\texttt{Oracle}(v) \neq a_j$, and a probability parameter $\delta \in (0,1)$.\footnote{When $m=2$, the loop at Line~\ref{alg:hyperplane_for} can be skipped, as the algorithm can simply employ $p^\circ$ to compute $H_{jk}$.}

\begin{wrapfigure}[38]{R}{0.53\textwidth}
\vspace{-0.8cm}
\begin{minipage}{0.53\textwidth}
\begin{algorithm}[H]
	\caption{\texttt{Find-Hyperplane}}
	\label{alg:hyperplane}
	\small
	\begin{algorithmic}[1]
		\Require $a_j, \mathcal{U}_{j}, p^\text{int}, v, \delta$ \hfill \LComment{\textnormal{As given by Algorithm~\ref{alg:learning_commitment}}}
		\State $p \gets \texttt{Sample-Int}(\mathcal{U}_{j}, \delta)$
		\State $p^1 \gets p$ \hfill \LComment{\textnormal{$p^1$ always selected at random}}
		\If{ $\texttt{Oracle} (p)=a_j$ }
		\State $p^2 \gets v $ \hfill \LComment{\textnormal{$v \notin \mathcal{P}_j$ by design}}
		\Else
		\State $p^2 \gets p^\text{int}$\hfill \LComment{\textnormal{$\texttt{Oracle} (p)\neq a_j$}}
		\EndIf
		\State $p^\circ \gets \texttt{Binary-Search}(a_j,p^1, p^2)$\label{line:binary_search_hp_1}
		\State $\alpha \gets 2^{-m(B+4L)-1} / m$ \hfill \LComment{\textnormal{$B = $ bit-complexity of $p^\circ$}}\label{line:alpha_def}
		\State $\mathcal{S}_j \gets \varnothing$; $\mathcal{S}_k \gets \varnothing$
		\For{$i = 1 \ldots m$} \label{alg:hyperplane_for}
		\State $ p \gets \texttt{Sample-Int}(H_i \cap \Delta_{m}, \delta)$\label{line:sample_from_facet} \hfill \LComment{\textnormal{Sample a facet}}
		\State ${p}^{+i} \gets p^\circ + \alpha (p-p^\circ)$ \label{line:point_simlex}	
		\State $p^{-i} \gets p^\circ - \alpha (p-p^\circ)$
		\If{ $\texttt{Oracle} ({p}^{+i})=a_j$ }
		\State $\mathcal{S}_{j} \gets \mathcal{S}_{j} \cup \{p^{+i} \}  \wedge  \mathcal{S}_{k} \gets \mathcal{S}_{k} \cup \{ p^{-i} \} $
		\Else
		\State $\mathcal{S}_{k} \gets \mathcal{S}_{k} \cup \{p^{+i} \}  \wedge  \mathcal{S}_{j} \gets \mathcal{S}_{j} \cup \{ p^{-i} \} $
		\EndIf
		\EndFor
		\State Build $H_{jk}$ by $\texttt{Binary-Search}(a_j,p^1, p^{2})$ for $m-1$ pairs of linearly-independent points $p^1 \in \mathcal{S}_j, p^2 \in \mathcal{S}_k$
	\end{algorithmic}
\end{algorithm}
\end{minipage}
\begin{minipage}{0.53\textwidth}
	\begin{figure}[H]
		\centering
		\includegraphics[width=0.6\linewidth]{Tikz_pictures/hyperplane_example.tex}
		\caption{Example of points $p^{+i}$ and $p^{-i}$ computed by Algorithm~\ref{alg:hyperplane}, with the $m-1$ line segments used to compute the separating hyperplane $H_{jk}$.}
		\label{fig:hyperplane_example}
	\end{figure}
\end{minipage}
\end{wrapfigure}

First, Algorithm~\ref{alg:hyperplane} samples a point $p$ in the interior of $\mathcal{U}_{j}$, using the \texttt{Sample-Int} procedure.
If $\texttt{Oracle}(p)$ is equal to $a_j$, it performs a binary search between the sampled point $p$ and the vertex $v$.
If $\texttt{Oracle}(p)$ returns an action different from $a_j$, it performs a binary search between  $p^\text{int}$ and $p$.
%
%
Selecting the point $p$ at random is fundamental.
This ensures that, with high probability, the segment on which the binary search is performed does \emph{not} intersect a point on a facet of the best-response region $\mathcal{P}_j$ where two or more separating hyperplanes $H_{jk}$ intersect.
%
%
Moreover, thanks to Lemma~\ref{lem:vertex_two} and $\texttt{Oracle}(v) \neq a_j$, action $a_j$ is \emph{never} a best response in $v$.
This excludes that the binary search is done on a segment fully contained in $\mathcal{P}_j$, an issue of~\citet{Peng2019} (see Figure~\ref{fig:simplex_1}).
%
%

After the two extremes $p^1, p^2$ of the segment have been identified, Algorithm~\ref{alg:hyperplane} calls the \texttt{Binary-Search} procedure (Line~\ref{line:binary_search_hp_1}).
The  latter computes a point $p^\circ \in \Delta_{m}$ on a new separating hyperplane $H_{jk}$ with $a_k \in {A}_f$.
However, to actually compute the hyperplane, the algorithm has to identify $m-1$ linearly-independent points on it.\footnote{Notice that $m-1$ linearly-independent points are sufficient since separating hyperplanes pass by the origin.}
To do so, Algorithm~\ref{alg:hyperplane} uses a method that is different from the one implemented in existing algorithms.
In particular, instead of randomly drawing a ``small'' $(m-1)$-dimensional simplex centered at $p^\circ$, our algorithm samples a point from the interior of each facet $H_i \cap \Delta_{m}$ of leader's strategy space, by calling the \texttt{Sample-Int} procedure (Line~\ref{line:sample_from_facet}).
%
Then, it takes suitable linear combinations between sampled points and $p^\circ$, defined by a parameter $\alpha$ (see Line~\ref{line:alpha_def}).
Specifically, the algorithm takes the points $p^{+i} \coloneqq p^\circ + \alpha (p - p^\circ)$, where $p$ is the point sampled from the facet $H_i \cap \Delta_{m}$.
Moreover, the algorithm also takes the points $p^{-i} \coloneqq p^\circ - \alpha (p - p^\circ)$, which represent ``mirrored'' versions of the points $p^{+i}$.
Algorithm~\ref{alg:hyperplane} splits points $p^{+i}$ and $p^{-i}$ into two sets $\mathcal{S}_j$ and~$\mathcal{S}_k$, depending on which follower's action is a best response in such points, either $a_j$ or $a_k$.
Considering the ``mirrored'' points $p^{-i}$ is needed to ensure that the algorithm takes at least one point in which follower's best response is $a_j$ and at least one point in which the best response is $a_k$.
%
Finally, Algorithm~\ref{alg:hyperplane} takes $m-1$ pairs of linearly-independent points, by taking one point from $\mathcal{S}_j$ and the other from $\mathcal{S}_k$, and it computes the coefficients of the separating hyperplane $H_{jk}$.
Figure~\ref{fig:hyperplane_example} shows an example of how Algorithm~\ref{alg:hyperplane} works.
%
%
%

The following lemma shows some fundamental properties of the points $p^{+i}$:
%
%
\begin{restatable}{lemma}{vectorDifferent}\label{lem:lin_indep}
	With probability at least $1 - 2 \delta m$, the points $p^{+i}$ computed by Algorithm~\ref{alg:hyperplane} belong to $\textnormal{int}(\Delta_m)$, are linearly independent, and are {not} on the hyperplane $H_{jk}$.
	%
\end{restatable}
Notice that Lemma~\ref{lem:lin_indep} is proved by employing the properties of the procedure \texttt{Sample-Int}, shown in Lemma~\ref{lem:sample_point_first}.
The following lemma shows that the sets $\mathcal{S}_j$ and $\mathcal{S}_k$ built by Algorithm~\ref{alg:hyperplane} are always well defined.
In particular, it shows that, in each point $p^{+i}$ or $p^{-i}$, either $a_j$ or $a_k$ is a best response for the follower, and, additionally, the best response in $p^{+i}$ is always different from that in $p^{-i}$.
Formally:
%
%
\begin{restatable}{lemma}{intersectionNull}\label{lem:hyper_ball}
	With probability $\geq 1 - \delta (m+n)^2$, for every $p^{+i}$ in Algorithm~\ref{alg:hyperplane}, $a^\star_f(p^{+i}) \in \{ a_j, a_k \}$. The same holds for $p^{-i}$.
	If $p^{+i} \not \in H_{jk}$ and $a^\star_f(p^{+i}) = a_j$, then $a^\star_f(p^{-i}) = a_k$, and viceversa.
	%
	%
\end{restatable}
Lemma~\ref{lem:hyper_ball} follows since there exists a neighborhood of $p^\circ$ that does \emph{not} intersect any other separating/boundary hyperplane defining $\mathcal{P}_j$ and $\mathcal{P}_k$, as the bit-complexity of all elements involved is bounded.
Moreover, with high probability, the segment on which binary search is performed does \emph{not} intersect hyperplanes different from $H_{jk}$.
%
%
%
By Lemma~\ref{lem:hyper_ball}, with high probability, the sets $\mathcal{S}_j$ and $\mathcal{S}_k$ are well defined, \emph{i.e.}, it is possible to identify $m-1$ pairs $p^1, p^2$ of linearly-independent points with $p^1 \in \mathcal{S}_j$ and $p^2 \in \mathcal{S}_k$.
%
%
Finally, by Lemmas~\ref{lem:lin_indep}~and~\ref{lem:hyper_ball}:
\begin{restatable}{lemma}{hyperSegmen}\label{lem:find_hyperplane}
	With probability of at least $1- 2(m+n)^2 \delta$, Algorithm~\ref{alg:hyperplane} returns a separating hyperplane $H_{jk}$ by using $\mathcal{O}(m^7L + m^4\log(\nicefrac{1}{\delta}))$ samples. 
	%
\end{restatable}
The number of samples required by Algorithm~\ref{alg:hyperplane} can be bounded by observing that binary search is always performed between points with bit-complexity of the order of $\mathcal{O}(m^5L+ m^2\log(\nicefrac{1}{\delta}))$, and that $\mathcal{O}(m)$ binary searches are performed.
This is ensured by Lemmas~\ref{lem:sample_point_first}~and~\ref{lem:binary_search}, and the fact that the bit-complexity of a vertex is $\mathcal{O}(m^2L)$.
%
%





\subsection{\texttt{Binary-Search}}\label{sec:binary_search}

The pseudocode of \texttt{Binary-Search} is provided in Algorithm~\ref{alg:binary_search}.
The algorithm performs a binary search on the line segment of extreme points $p^1, p^2$, with $\texttt{Oracle}(p^1) = a_j$ and $\texttt{Oracle}(p^2) \neq a_j$.
At each iteration of the binary search, the algorithm queries the middle point of the line segment.
If the best response played by the follower is $a_j$, then the algorithm advances one extreme of the line segment, otherwise it moves the other extreme.
The binary search is done until the line segment is sufficiently small, so that a point on a new separating hyperplane can be computed.

\begin{wrapfigure}[12]{R}{0.55\textwidth}
	\vspace{-.8cm}
\begin{minipage}{0.55\textwidth}
\begin{algorithm}[H]
	\caption{\texttt{Binary-Search}}\label{alg:binary_search}
	\small
	\begin{algorithmic}[1]
		\Require $a_j$, $p^1, p^2 $ of bit-complexity $\leq B$\hfill \LComment{\textnormal{Algorithm~\ref{alg:hyperplane}}}
		\State $\lambda_1 \gets 0 $; $\lambda_2 \gets 1 $
		\While{$|\lambda_2 - \lambda_1 | \ge 2^{-6m(5B+8L)}$}
		\State $\lambda \gets {(\lambda_1 + \lambda_2)}/{2} $; $p^\circ \gets p^1 + \lambda (p^2 - p^1)$
		\If{ $\texttt{Oracle}(p^\circ) = a_j$ }
		\State $\lambda_1 \gets \lambda $ \hfill \LComment{\textnormal{New point with $a_j$ as best response}}
		\Else
		\State $\lambda_2 \gets \lambda $ \hfill \LComment{\textnormal{New point with best response $\neq a_j$}}
		\EndIf
		\EndWhile
		\State $\lambda \gets \texttt{Stern-Brocot-Tree}(\lambda_1, \lambda_2,3m(5B +8L))$
		\State $p^\circ \gets \lambda p^1 + (1- \lambda) p^2$
	\end{algorithmic}
\end{algorithm}
\end{minipage}
\end{wrapfigure}

The number of binary search steps needed to determine a point on an hyperplane depends on the bit-complexity $B$ of $p^1,p^2$, and the bit-complexity $L$ of follower's payoffs.
%
%
Indeed, in an interval with a suitably-defined length there exists a single rational number with bit-complexity $\mathcal{O}(m(B+L) )$, and this is the point that belongs to the separating hyperplane.
Such a point can be efficiently computed by binary search on the Stern–Brocot tree.
See~\citep{Fori2007} for an efficient implementation.
%
%
\begin{restatable}{lemma}{binarySearch}\label{lem:binary_search}
	Let $p^1, p^2 \in \Delta_{m}$ be such that $a^\star_f(p^1) = a_j \ne a^\star_f(p^2)$ with $a_j \in {A}_f$, they are not both in $\mathcal{P}_j \setminus \textnormal{int} (\mathcal{P}_j)$, and they have bit-complexity bounded by $B$.
	Then, Algorithm~\ref{alg:binary_search} finds a point $p^\circ \in H_{jk}$ for some $a_k \in {A}_f$.
	Furthermore, the algorithm requires $\mathcal{O}(m(B + L))$ samples, and the point $p^\circ$ has bit-complexity bounded by $\mathcal{O}(m(B + L))$.
	%
	%
\end{restatable}

\subsection{\texttt{Sample-Int}}\label{sec:sample point}

\begin{wrapfigure}[17]{R}{0.53\textwidth}
	\vspace{-0.85cm}
	\begin{minipage}{0.53\textwidth}
\begin{algorithm}[H]
	\caption{\texttt{Sample-Int}}
	\label{alg:sample_point}
	\small
	\begin{algorithmic}[1]
		\Require $\mathcal{P} \subseteq \Delta_m: \textnormal{vol}_{m-1}(\mathcal{P})>0 \vee \mathcal{P} := H_i \cap \Delta_{m}, \delta$
		\State $\left\{ \begin{array}{ll}
		d \gets m & \textcolor{gray}{\triangleright\textnormal{ If }\mathcal{P} \subseteq \Delta_m:\textnormal{vol}_{m-1}(\mathcal{P})>0} \\
		d \gets m-1 & \textcolor{gray}{\triangleright\textnormal{ If }\mathcal{P} = H_i \cap \Delta_{m}}
		\end{array}\right.$
		\State $\mathcal{V} \gets d$ linearly-independent vertexes of $\mathcal{P}$\label{line:sample_li_vertices}
		\State $p^\diamond \gets \frac{1}{d} \sum_{v \in \mathcal{V}} v$\label{line:sample_1}
		\State $\rho \gets \left(d^3 2^{9d^3L +4dL} \right)^{-1}$; $M \gets \left \lceil \nicefrac{\sqrt{d}}{\delta} \right \rceil$\label{line:sample_rho} 
		\State $x \sim \text{Uniform}(\{- 1, -\frac{M-1}{M}, \ldots, 0, \ldots, \frac{M-1}{M}, 1 \}^{d-1})$\label{line:sample_2}
		\State $l \gets 1$; $k \gets 1$
		\While{$l \leq d-1$}
		\State \textbf{if} $k \neq i$ \textbf{then} $p_k \gets p_k^\diamond + \rho x_l $; $l \gets l+1$
		\State \textbf{else} $p_k \gets 0$; 
		\State $k \gets k+1$
		\EndWhile
		\State \textbf{if} $i = m$ \textbf{then} $p_{m-1} \gets 1-\sum_{l= 1}^{m-2} p_l$
		\State \textbf{else} $p_{m} \gets 1-\sum_{l= 1}^{m-1} p_l$ 
	\end{algorithmic}
\end{algorithm}
\end{minipage}
\end{wrapfigure}

The pseudocode of \texttt{Sample-Int} is in Algorithm~\ref{alg:sample_point}. It takes as input a polytope $\mathcal{P} \subseteq \Delta_m:\text{vol}_{m-1}(\mathcal{P}) > 0$ or a facet $\mathcal{P} \coloneqq H_i \cap \Delta_{m}$ and a probability parameter $\delta \in (0,1)$, and it returns a point from $\textnormal{int}(\mathcal{P})$.
First, Algorithm~\ref{alg:sample_point} computes $p^\diamond \in \text{int}(\mathcal{P})$ by taking the average of $d$ linearly-independent vertexes of $\mathcal{P}$ (Line~\ref{line:sample_1}), where either $d=m$ or $d=m-1$, depending on the received input.
Then, the algorithm samples a point $x \in \mathbb{R}^{d-1}$ from a discrete uniform distribution with support on points belonging to an hypercube (Line~\ref{line:sample_2}).
In particular, the support is an equally-spaced grid in the hypercube, with step $1/M$ (Line~\ref{line:sample_rho}).
Then, the algorithm adds each component $x_l$ of $x$ to a component $p^\diamond_k$ of $p^\diamond$, by scaling it by a factor $\rho$ (Line~\ref{line:sample_rho})
The last component of $p$ is set so that the $p_i$ sum to $1$.
Notice that the choice of $\rho$ is done to guarantee that the $p$ produced by Algorithm~\ref{alg:sample_point} belongs to $\text{int}(\mathcal{P})$.
The main property of Algorithm~\ref{alg:sample_point} is that, given any linear space $H \subseteq \mathbb{R}^{m}$ (with $\mathcal{P} \not\subseteq H$), the probability that the sampled point belongs to $H$ is at most $\delta$.
%
%
Notice that, when existing algorithms~\citep{Peng2019,Letchford2009} randomly sample a point, they consider the probability of the aforementioned event to be zero, since they use strategies with arbitrarily large bit-complexity.
The property is made formal by the following lemma.
\begin{restatable}{lemma}{samplepointfirst}
	\label{lem:sample_point_first}
	Given a polytope $\mathcal{P} \subseteq \Delta_m:\textnormal{vol}_{m-1}(\mathcal{P}) >~0$ defined by separating or boundary hyperplanes, or a facet $\mathcal{P} \coloneqq H_i \cap \Delta_{m}$, Algorithm~\ref{alg:sample_point} computes $p \in \textnormal{int}(\mathcal{P})$ such that, for every linear space $H \subset \mathbb{R}^m : \mathcal{P} \not\subseteq H$ of dimension at most $m-1$, the probability that $p \in H$ is at most $\delta$.
	Furthermore, the bit-complexity of $p$ is $\leq 40m^3L + 2\log_2(\nicefrac{1}{\delta})$.
	%
\end{restatable}
%
%
%
Lemma~\ref{lem:sample_point_first} follows by bounding the volume of the intersection of the hypercube with a linear space of dimension at most $m-1$~\citep{Ball1986,Ivano2021}, and by the fact that the bit-complexity of all elements involved is bounded.

\clearpage

\subsubsection*{Acknowledgments}

This paper is supported by the Italian MIUR PRIN 2022 Project “Targeted Learning Dynamics: Computing Efficient and Fair Equilibria through No-Regret Algorithms”, by the FAIR (Future Artificial Intelligence Research) project, funded by the NextGenerationEU program within the PNRR-PE-AI scheme (M4C2, Investment 1.3, Line on Artificial Intelligence), and by the EU Horizon project ELIAS (European Lighthouse of AI for Sustainability, No. 101120237).

%

\bibliographystyle{plainnat}
\bibliography{references}

\appendix
\newpage
\onecolumn
\section*{Appendix}
The appendixes are organized as follows:
\begin{itemize}
	\item Appendix~\ref{sec:related_works} provides a detailed discussion of the previous works most related to ours.
	\item Appendix~\ref{sec:problem_prev_works} provides a detailed discussion on instances in which the algorithmic approaches by~\citet{Letchford2009} and~\citet{Peng2019} fail.
	\item Appendix~\ref{sec:appendix_omitted_proof_main} provides all the proofs omitted from the main body of the paper.
	\item Appendix~\ref{sec:omitted_lemma_proof} discusses some useful technical details omitted from the paper.
	\item Appendix~\ref{sec:app_running_time} provides a discussion on the running time of Algorithm~\ref{alg:learning_commitment}.
	\item Appendix~\ref{sec:app_extension} provides a discussion on how the results presented in the main body of the paper generalize to the case in which there are equivalent follower's actions.
\end{itemize}

\section{Related Works}\label{sec:related_works}
\paragraph{Learning in Stackelberg games} The two works closely related to ours are the ones by~\cite{Letchford2009} ~and~\citet{Peng2019}, as discussed in Section~\ref{sec:sota}. In addition to these two works, the problem of learning optimal strategies in Stackelberg games has received significant attention from the scientific community. 
While we examined a setting with a finite number of actions and where the leader can commit to mixed strategies, \citet{Fiez2020}~study games with continuous action spaces where the leader commits to a single action. They present an algorithm based on a gradient-descent approach. 
\citet{bai2021sample}~propose a model where both the leader and the follower learn through repeated interactions, observing only noisy samples of their rewards. 
\citet{lauffer2022noregret}~study Stackelberg games characterized by a random state that influences the leader's utility and the set of available actions.

\paragraph{Learning in Stackelberg security games} 
Our work is also related to the problem of learning optimal strategies in Stackelberg security games. 
In these settings a defender needs to allocate a set of resources to protect some targets against one or multiple attackers. 
\citet{Blum2014learning}~provide an algorithm to learn an approximately optimal strategies in polynomial time in the number of targets. 
\citet{Balcan2015Commit}~consider instead a sequence of attackers of different types. 
In particular, they focus on the case in which an adversary chooses the attacker's types from a known set and they design a no-regret algorithm. 
\citet{Haghtalab2016}~focuses instead on security games where the attacker is not completely rational.
Thanks to this assumption, they develop an algorithm that can be executed offline on an existing dataset.
Finally, the algorithm proposed by \citet{Peng2019} for general Stackelberg games can be modified to tackle directly Stackelberg security games in their natural representation.


\section{Issues of the algorithms by~\citet{Peng2019}~and~\citet{Letchford2009}}\label{sec:problem_prev_works}


\subsection{Minor issue of the algorithm by~\citet{Peng2019}}


In the following, we present an instance in which the algorithm proposed by~\citet{Peng2019} is not guaranteed to terminate correctly. Intuitively, the procedure proposed by~\cite{Peng2019} maintains an upper bound for each best response region of the follower's action it has already discovered and proceeds by shrinking them by discovering new separating hyperplanes. One of the main differences with the approach by~\cite{Letchford2009} is that, whenever the intersection of the upper bounds of the discovered follower's actions has zero measure, the algorithm enumerates the vertices of each upper bound and either terminates or discovers a new follower's action. Meanwhile, the algorithm also maintains a lower bound for the best response regions of the follower's action it has discovered so far.

We consider the instance presented in Figure~\ref{fig:peng_fails_1}, in which $|\mathcal{A}_\ell|=|\mathcal{A}_f|=3$. As a first step, the algorithm proposed by~\cite{Peng2019} prescribes to randomly sample a point from the interior of $\Delta_m$, which coincides with the commitment $p \in \mathcal{P}_{3}$ in Figure~\ref{fig:peng_fails_1}. Then, it initializes the upper bound $\mathcal{U}_{3}$ with the whole simplex and queries each vertex of such upper bound. When the algorithm queries the vertex $v^1$, it also chooses one of the two followers' best responses to eventually add the commitment $v^1$ to a suitably defined lower bound. More specifically, the algorithm sets the lower bound $\mathcal{L}_{1} = \{v^1\}$ and also updates its corresponding upper bound $\mathcal{U}_{1} = \Delta_3$. Equivalently, the algorithm updates the lower bound of the follower's action $a_3$, \emph{i.e}., $\mathcal{L}_{3} = \text{co}(v^2, v^3, p)$. As a further step, the algorithm attempts to find a separating hyperplane between the upper bounds $\mathcal{U}_{1}$ and $\mathcal{U}_{3}$. To do so, it follows the procedure proposed by~\cite{Letchford2009} and deterministically computes the point $p^1$ in Figure~\ref{fig:peng_fails_1} lying at the intersection of the two upper bounds $\mathcal{U}_{1}$ and $\mathcal{U}_{3} = \Delta_3$. Thus, following the steps presented in~\cite{Letchford2009}, the procedure performs a binary search between the commitment $p^1$ and the vertex $v^1$ since $v^1 \in \mathcal{L}_{1}$ and $a_f^\star(p^1)=a_2$. If the binary search correctly terminates, it returns the vertex $v^1$ as the point ling on the hyperplane $H_{12}$. Consequently, at least one vertex of the small simplex centred in $v^1$ falls outside $\Delta_3$ and cannot be queried, leading to a failure of the procedure by~\cite{Peng2019}.

\begin{figure*}[!htp]
	\centering
	\includegraphics[width=0.36\linewidth]{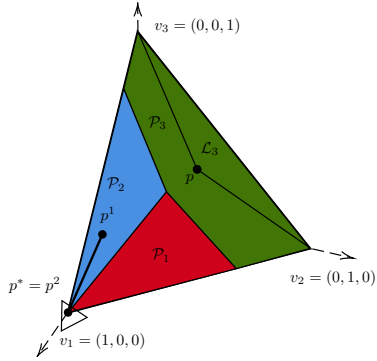}
	\caption{Instance in which the approach by~\cite{Peng2019} fails due to how $p^2$ is selected}
	\label{fig:peng_fails_1}
\end{figure*}

\subsection{The algorithm by~\citet{Peng2019} may require exponentially-many samples}

In the following, we present an SG in which the binary search procedure proposed by by~\citet{Letchford2009} and employed by the algorithm by~\citet{Peng2019} to compute a separating hyperplane requires an exponential number of samples, if the leader's strategies are \emph{not} specified with a bounded bit-complexity.
We consider the scenario depicted in Figure~\ref{fig:ubounded_binary_search_app}. This refers to an SG in which both the leader and the follower have three actions available, namely $A_\ell = \{a_1,a_2,a_3 \}$ and $A_f = \{a_1,a_2,a_3 \}$. Furthermore, follower's utilities are specified by the following matrix:
$$
u_f(a_i,a_j) = \left( \begin{matrix}
	0 & 1 & 0\\
	1 & 0 & 0\\
	0 & 0 & 1
\end{matrix}\right)$$
where the entry at position $(i,j)$ encodes follower's payoff $u_f(a_i,a_j)$ when the leader plays action $a_i \in {A}_{\ell}$ and the follower plays action $a_j \in {A}_{f} $.
The separating hyperplanes that define the three follower's best-response regions are given by: 
\begin{align*}
\begin{cases}
	H_{12}:=\{p \in \mathbb{R}^m \mid p_1 - p_2 = 0\}, \\
	H_{13}:=\{p \in \mathbb{R}^m \mid p_2 - p_3 = 0\}, \\
	H_{23}:=\{p \in \mathbb{R}^m \mid p_1 - p_3 = 0\}.
\end{cases}
\end{align*}
%

\begin{figure*}[!htp]
	\centering
	\includegraphics[width=0.32\linewidth]{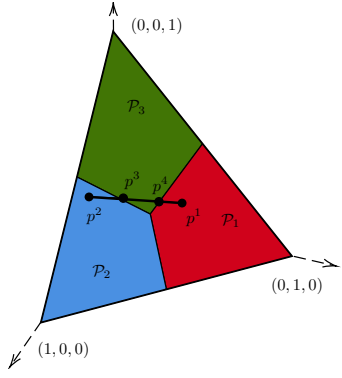}
	\caption{The binary search proposed employed by~\citet{Peng2019} requires exponential samples.}
	\label{fig:ubounded_binary_search_app}
\end{figure*}

The algorithm by~\citet{Peng2019} starts by sampling a point at random from the simplex that coincides with $p^1 \in \Delta_3$ in the example depicted in Figure~\ref{fig:ubounded_binary_search_app}. Then, the algorithm by~\citet{Peng2019} observes the follower's best-response in $p^1$, which coincides with the follower's action $a_1$.
Consequently, it initializes the relative upper bound $\mathcal{U}_1$ to the whole simplex $\Delta_m$, and the lower bound $\mathcal{L}_1$ to the set $\{p^1\}$.
As a further step, the algorithm checks the vertex $v^1=(1,0,0)$ of $\mathcal{U}_1=\Delta_m$.
Since the follower's best-response in this point is $a_2$, it initializes the upper bound $\mathcal{U}_2$ to $\Delta_m$ and the the lower bound $\mathcal{L}_2$ to the set $\{v^1\}$.
Subsequently, the algorithm by~\citet{Peng2019} computes a hyperplane defining a facet of $\mathcal{P}_1$.

To do so, it computes a point in the interior of intersection of the two upper bounds of the follower's best-response regions $\mathcal{P}_1$ and $\mathcal{P}_2$, which initially coincides with the whole simplex. If the sampled point coincides with $p^2$ in Figure~\ref{fig:ubounded_binary_search_app}, then the algorithm performs binary search between $p^1$ and $p^2$, since the follower's best response in $p^2$ coincides with $a_2$. 
%
Formally, the commitments $p^1$ and $p^2$ presented in Figure~\ref{fig:ubounded_binary_search_app} are defined as $ p^1 \coloneqq \left(\nicefrac{1}{3}-\epsilon,\nicefrac{1}{3}+\epsilon,\nicefrac{1}{3} \right) \in \mathcal{P}_{1}$ and $p^2=\left(\nicefrac{1}{2},\nicefrac{1}{10},\nicefrac{2}{5} \right) \in \mathcal{P}_{2}$, for some constant $\epsilon >0$.
Clearly, since the commitment $p^1$ is sampled at random from the simplex, a large number of bits may be needed, as $\epsilon$ could potentially be arbitrarily.
We also define $p^3 \in \Delta_m$ as the point at the intersection of the line segment connecting $p^1$ to $p^2$ and the hyperplane $H_{23}$.
Therefore, there exists a single value of $\lambda' \in (0,1)$ such that $p^3 = \lambda' p^1 + (1-\lambda')p^2$.
%
Simple calculations show that the value of such a $\lambda'$ is $({10\epsilon +1})^{-1}$, which in turn implies that $p^3 = {(30\epsilon +3)}^{-1} ( 12\epsilon+1, 6\epsilon+1, 12\epsilon+1)$. 

A crucial observation is that, to compute the point $p^4 \in \Delta_m$ at the intersection between the line segment that connects $p^1$ to $p^2$ and the hyperplane $H_{13}$, the binary search must employ at least $\mathcal{O}(\log_2(\nicefrac{d_2(p^1,p^2)}{d_2(p^1,p^3)}))$ samples, where we let $d_2(p^1,p^2)\coloneqq \| p^1-p^2\|_2$ and $d_2(p^1,p^3)\coloneqq \| p^1-p^3\|_2$. This is because, to correctly identify the point $p^4 \in H_{13} \cap \text{co}(p^1,p^2)$, the binary search has to query at least one point that belongs to the line segment connecting $p^1$ to $p^3$, as $p^4 \in  \text{co}(p^1,p^3)$. Thus, by observing that $p^1 -p^3 = {3(10\epsilon+1)}^{-1} (-30\epsilon^2 -5\epsilon, 30\epsilon^2 +7\epsilon, -2\epsilon)$, we can compute the two square distances $d_2(p^1,p^2)^2$ and $d_2(p^1,p^3)^2$ as follows.
\begin{equation*}
	d_2(p^1,p^2)^2 = \sum_{i \in [3]} (p^1_i-p^2_i)^2= \left(-\frac{1}{6}-\epsilon\right)^2 +  \left(\frac{7}{30}\right)^2 + \left(\frac{1}{15}\right)^2
	= 2\epsilon^2 +\frac{4}{5}\epsilon +\frac{13}{150}.
\end{equation*}

Similarly, the square distance $d_2(p^1,p^3)^2$ is equal to:
\begin{align*}
	d_2(p^1,p^3)^2 &= \sum_{i \in [3]} (p^1_i-p^3_i)^2 \\
	&= \frac{1}{9(10\epsilon+1)^2} \left((30\epsilon^2+5\epsilon)^2 + (30\epsilon^2+7\epsilon)+ 4\epsilon^2 \right)	 \\
	&= \frac{1800\epsilon^4 +720\epsilon^3 +78\epsilon^2}{9(100\epsilon^2 +20\epsilon +1)}
\end{align*}
As a result, the ratio between such distances is given by:
\begin{align*}
	\frac{d_2(p^1,p^2)}{d_2(p^1,p^3)} = \sqrt{\frac{9(2\epsilon^2 +\frac{4}{5}\epsilon +\frac{13}{150})(100\epsilon^2+20\epsilon+1)}{1800\epsilon^4 +720\epsilon^3 +78\epsilon^2}} = \mathcal{O}\left(\frac{1}{\epsilon}\right).
\end{align*}

As previously observed, the number of samples required by the binary search to be correctly executed is of the order $\mathcal{O}(\log_2(\nicefrac{d_2(p^1,p^2)}{d_2(p^1,p^3)}))$ and, consequently, $\mathcal{O}(\log_2(\nicefrac{1}{\epsilon}))$. Therefore, when $\epsilon=2^{-2^L}$, the sample complexity of such a procedure is $\mathcal{O}(2^L)$, showing an exponential dependence in the parameter $L$. Intuitively, this is because the point $p^1$ is taken arbitrarily close to the commitment $p=(\nicefrac{1}{3},\nicefrac{1}{3},\nicefrac{1}{3}),$ which lies at the intersection of the three best response regions. As a consequence, it is necessary to bound the bit complexity of the leader's commitments to avoid an exponential dependence on the number of samples required to compute an optimal solution.\footnote{Given $n \in \mathbb{N}_{+}$, we let $[n]\coloneqq \{1,\dots,n\}$.}

\subsection{Minor issue in the procedure by~\citet{Letchford2009}}

In this section, we present an instance in which the algorithm proposed by~\cite{Letchford2009} and employed by~\cite{Peng2019} to compute a separating hyperplane does not terminate correctly.
Given two overestimates $\mathcal{U}_j$ and $\mathcal{U}_k$ of some best-response regions $\mathcal{P}_j$ and $\mathcal{P}_k$ such that $\text{vol}(\mathcal{U}_j \cap \mathcal{U}_k)>0$, this procedure finds a separating hyperplane $H_{jh}$ for some $h \in [h]$.
Observe that $\mathcal{P}_j$ and $\mathcal{P}_k$ may not share a facet, and thus $a_h$ can be different from $a_k$.
To identify this hyperplane, the algorithm deterministically computes a point $p^1 \in \text{int}(\mathcal{U}_j \cap \mathcal{U}_k)$.
Subsequently, it performs a binary search on the segment connecting $p^1$ to a suitable point $p^2$, where $p^2$ belongs to either $\mathcal{L}_j$ or $\mathcal{L}_k$ depending on the follower's best response in $p^1$.
Then, a random simplex of dimension $m$ is built centered at the point returned by the former procedure, and new binary searches are performed between the vertices of such a simplex in which different actions are implemented. 
In this way, the algorithm eventually computes a set of linearly independent points lying on the hyperplane, which is enough to uniquely identify it.

\begin{figure*}[!htp]
	\centering
	\includegraphics[width=0.32\linewidth]{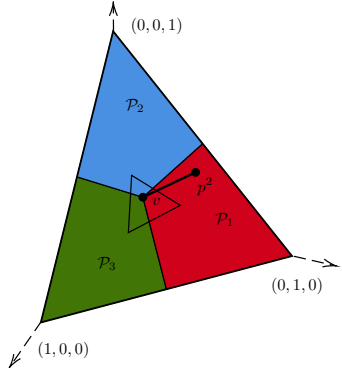}
	\caption{Instance in which the procedure by~\cite{Letchford2009} to compute hyperplanes fails}
	\label{fig:letchford_fails}
\end{figure*}

The issue with the algorithm proposed by~\cite{Letchford2009} lies in the choice of the segment where the binary search is executed. 
Specifically, we observe that the point $p^1$ should \emph{not} be chosen deterministically.
In the following, we consider an instance (see Figure~\ref{fig:letchford_fails}) where $|\mathcal{A}_\ell|=|\mathcal{A}_f|=3$. 
Furthermore, we assume that the three follower's best-response regions share a vertex, \emph{i.e.}, $\mathcal{P}_{1} \cap \mathcal{P}_{2} \cap \mathcal{P}_{3} = {v}$ and the action actually played by the follower in such a commitment is equal to $a_3$, namely $a_f^\star(v)=a_3$.  
The algorithm by~\citet{Peng2019} samples a random point from the simplex, which we suppose to be the point $p^2$ in Figure~\ref{fig:letchford_fails}.
Since $a^\star_f(p^2)=a_1$, the algorithm initializes the upper bound $\mathcal{U}_1$ of $\mathcal{P}_1$ to the simplex $\Delta_3$, and the corresponding lower bound $\mathcal{L}_1$ to $\{p^1\}$.
Subsequently, it queries the vertex $v^1=(1,0,0)$ of $\mathcal{U}_1$.
Thus, it discovers the follower's action $a_3$ and initializes $\mathcal{U}_3 = \Delta_3$ and $\mathcal{L}_3 = \{a_3\}$.

At this point, the Algorithm by~\citet{Peng2019} employs the procedure proposed by~\citet{Letchford2009} to compute a new separating hyperplane.
In particular, it computes deterministically a point $p^1$ in the interior of $\mathcal{U}_1 \cap \mathcal{U}_3=\Delta_3$.
If such a point is the vertex $v$ at the intersection of the three best-response regions, then the algorithm performs a binary search on the segment connecting $p^1=v$ and $p^2 \in \mathcal{L}_1$, as $a^\star_f(v)=a_3 \neq a_1$.
Thus, the point $v$ coincides with the commitment returned by such a procedure, which belongs to the hyperplane $H_{13}$. 
As a further step, the procedure by~\cite{Letchford2009} attempts to organize the vertices of a small simplex centered in $v$ into two sets in which either $a_1$ or $a_3$ is a best response. 
However, such a procedure fails no matter the how small the simplex is drawn, as its three vertices lie in three (and not two) different best-response regions. 
This example shows that if the point $p^1$ over which the binary search is performed is chosen deterministically, there may be instances in which the procedure proposed by~\cite{Letchford2009} fails with probability one.

\section{Omitted Proofs}\label{sec:appendix_omitted_proof_main}


\subsection{Proofs Omitted from \texttt{Main-Algorithm}}
\vertextwo*
\begin{proof}
($\Leftarrow$) If $p' \in \preg_{j}$, then $\widetilde{p}$ is also in $\preg_{j}$ by convexity of $\preg_{j}$.

($\Rightarrow$) In order to prove that $\widetilde{p} \in \preg_{j}$ implies $p' \in \preg_{j}$, we assume by contradiction that $p' \not \in \mathcal{P}_{j}$.
This immediately implies that $p' \in \text{int}(\mathcal{H}_{kj} \cap \Delta_m)$ for some follower's action $a_k \in A_f : a_k \ne a_j$, otherwise we would have $p' \in \mathcal{P}_j$ (unless $\mathcal{P}_j = \Delta_m$, but in that case the contradiction is readily obtained).
%
%
Let $p^\circ \in \Delta_m$ be a leader's strategy defined as $p^\circ \coloneqq \text{co}(p, p') \cap H_{kj}$. Intuitively, $p^\circ$ coincides with the intersection between the separating hyperplane $H_{kj}$ and the line segment connecting the two points $p$ and $p'$.
By definition of $p^\circ$, there must exist a weight $\lambda^\circ \in (0,1)$ such that $p^\circ = \lambda^\circ p + (1- \lambda^\circ)p'$, since $p^\circ \in \text{co}(p,p')$. Notice that $\lambda^\circ \neq 0$, otherwise we would have $p' \in H_{kj}$, contradicting the assumption that $p' \not \in \preg_{j}$. Furthermore, the value of $\lambda^\circ$ can be easily computed as follows:
\begin{equation*}
		\lambda^\circ := \frac{\left|  \displaystyle\sum\limits_{a_i \in A_\ell} p'_i \Big(u_f(a_i,a_j) -u_f(a_i,a_k) \Big) \right | }{\left | \displaystyle\sum\limits_{a_i \in A_\ell} (p_i- p'_i) \Big(u_f(a_i,a_j) -u_f(a_i,a_k) \Big) \right|}.
\end{equation*}
Notice that the numerator in the definition of $\lambda^\circ$ is the sum of $m$ terms with bit-complexity bounded by $B+4L$, as a consequence of Lemma~\ref{lem:sum_rationals} in Appendix~\ref{sec:omitted_lemma_proof}. Thus, according to the lemma, the absolute value of their sum is lower bounded by $2^{-m(B+4L)}$, since $\lambda^\circ \neq 0$. Moreover, the denominator of $\lambda^\circ$ can be upper bounded as follows:
\begin{align*}
& \left|\sum_{i=1}^{m} (p_i- p'_i)  \Big(u_f(a_i,a_j) -u_f(a_i,a_k)\Big) \right | \le \\   & \quad \quad \quad  \left|\sum_{i=1}^{m}  p_i \Big(u_f(a_i,a_j) -u_f(a_i,a_k) \Big)\right| + \left|\sum_{i=1}^{m} p'_i \Big(u_f(a_i,a_j) -u_f(a_i,a_k)\Big)\right| \le 2.
\end{align*}
Thus, the value of $\lambda^\circ$ can be lower bounded by $ {2^{-m(B+4L)-1}}$.
Therefore, by observing that $\lambda < \lambda^\circ$, we have $\widetilde{p} \in \text{co}(p^\circ, p')$, $\widetilde{p} \neq p^\circ$, and $\widetilde{p} \neq p'$. Thus, $\widetilde{p} \notin \preg_{j}$, since $\widetilde{p} \in \text{int}(\mathcal{H}_{kj} \cap \Delta_m)$.
%
This is the contradiction which concludes the proof.
\end{proof}

\VertexStackelberg*
\begin{proof}
	Let $p^\star \in \Delta_m$ be an optimal strategy to commit to.
	
	First, let us consider the case in which $a_f^\star(p^\star) = a_j$ for some follower's actions $a_j \in A_f$ with $\text{vol}(\mathcal{P}_j)>0$. In such a case, the statement trivially holds since leader's expected utility is linear in $p$ over the polytope $\mathcal{P}_j$, and, thus, there must exist an optimal commitment (possibly different from $p^\star$) coinciding with one of the vertices of $\mathcal{P}_j$.
	
	Now, let us consider the case in which $a_f^\star(p^\star) = a_k$ for some follower's action $a_k \in A_f$ with $\textnormal{vol}(\mathcal{P}_k) = 0$.
	As a first step, we show that $p^\star$ is on the boundary (relative to $\Delta_m$) of some best-response region $\mathcal{P}_j$ with $\textnormal{vol}(\mathcal{P}_j) > 0$.
	Suppose by contradiction that $p^\star \in \textnormal{int}(\mathcal{P}_j)$.
	Then, $a_k$ is a best-response in $\mathcal{P}_j \cap \mathcal{H}_{kj}$, and it is easy to see that this set has positive volume. Hence, we reach a contradiction with $\textnormal{vol}(\mathcal{P}_k) = 0$.
	
	%
	%
	%
	As a consequence, we immediately get that the point $p^\star$ must belong to the boundary (relative to $\Delta_m$) of some $\mathcal{P}_j$ with $\textnormal{vol}(\mathcal{P}_{j}) > 0$.
	Moreover, $p^\star$ must also belong to the separating hyperplane $H_{jk}$, since $ \{a_k,a_j\} \subseteq A(p^\star)$.
	As a result, $p^\star \in \mathcal{P}_j \cap H_{jk}$.
	%
	%
	%
	Moreover, we have that $\mathcal{P}_j \cap H_{jk} \subseteq \mathcal{P}_k$. This holds since in any point in the set $\mathcal{P}_j \cap H_{jk}$ it must be the case that $a_j$ is a best response, and the utility of actions $a_j$ is equal to the one of action $a_k$ by definition of $H_{jk}$.
	
	Finally, it is easy to see that  $V(\mathcal{P}_j \cap H_{jk} ) \subseteq V(\mathcal{P}_j )$ by the definition of $\mathcal{P}_j$.
	By the convexity of $\mathcal{P}_j \cap H_{jk}$ and the fact that leader's expected utility is linear over $\mathcal{P}_j \cap H_{jk}$ when follower's best response is fixed to $a_f^\star(p^\star) = a_k$, we can conclude that there exists at least one vertex $p$ of $\mathcal{P}_j \cap H_{jk}$, such that the utility of the leader is equal to the one in $p^\star$ assuming that follower's best response is $a_f^\star(p^\star) = a_k$. The proof is concluded noticing that $a_k\in A(p)$ belongs to the set of best responses, and that ties are broken in favor of the leader.
	
\end{proof}

\maintheorem*
\begin{proof}
	
	First, we prove that, with probability at least $1- \delta n^2 (2(m+n)^2 + n)$, Algorithm~\ref{alg:learning_commitment} returns an optimal commitment in a finite number of steps. To do this, we employ an inductive argument on the different executions of the while loop at Line~\ref{line:partition_loop3}. Formally, we show that, if before the execution of the while loop at Line~\ref{line:partition_loop3} in Algorithm~\ref{alg:learning_commitment} it holds $\textnormal{vol}(\mathcal{P}_{j}) > 0$ and $\mathcal{U}_j = \preg_j$ for every action $a_j \in \mathcal{C}$, then such a while loop terminates in a finite number of rounds. Furthermore, the set $\mathcal{C}$ is updated as $\mathcal{C} \gets \mathcal{C} \cup \{a_k\}$, where $\textnormal{vol}(\mathcal{P}_{k}) > 0$ and $\mathcal{U}_k = \preg_k$, with probability at least of $1- \delta n (2(m+n)^2 + n ) $.
	
	We start by noticing that, with probability at least $1-\delta n^2$, the new action $a_k \in {A}_f$ added to $\mathcal{C}$ is such that $\textnormal{vol}(\mathcal{P}_{k}) > 0$. This is because the commitment $p^\textnormal{int}$ is sampled outside the union of the best-response regions $\mathcal{P}_j $ with $a_j \in \mathcal{C}$. Indeed, for every action $a_j \in \mathcal{C}$, it holds that $\mathcal{P}_{j} = \mathcal{U}_{j}$, thanks to the inductive hypothesis. As a result, $a^\star_f(p^{\textnormal{int}})\ne a_j$ for every $a_j \in \mathcal{C}$. Furthermore, with probability at least $1 - \delta n^2$, the strategy $p^\textnormal{int}$ does not belong to any separating hyperplane, as guaranteed by Lemma~\ref{lem:sample_point_first} and a union bound on the total number of separating hyperplanes, which is at most $n^2$. Consequently, $p^{\textnormal{int}} \in \textnormal{int}(\mathcal{P}_k)$ for some $a_k \in {A}_f$, proving that $\textnormal{vol}(\mathcal{P}_k)>0$.

	As a further step, we prove that if $p^{\textnormal{int}} \in \textnormal{int}(\mathcal{P}_k)$ with $\textnormal{vol}(\mathcal{P}_k)>0$, then with probability at least $1- 2\delta n(m+n)^2$, the while loop at Line~\ref{line:partition_loop3} terminates in a finite number of rounds. This is because, with probability at least $1- 2\delta n(m+n)^2$, during the execution of such a loop, the algorithm never computes the same separating hyperplane $H_{kl}$ with $a_l  \in A_\ell$ multiple times.
	Clearly, this hold only in the case, whenever invoked, Algorithm~\ref{alg:hyperplane} correctly identifies the separating hyperplane. Observe that during the execution of the while loop at Line~\ref{line:partition_loop3}, Algorithm~\ref{alg:hyperplane} is always invoked receiving as input a commitment $p^{\textnormal{int}} \in \textnormal{int}(\mathcal{P}_k)$ and a vertex $v$ in which $a^\star_f(v) \ne a_j$, as ensured by Lemma~\ref{lem:vertex_two}.
	Furthermore, the point $p^\circ$ found by binary search at Line~\ref{line:binary_search_hp_1} in Algorithm~\ref{alg:hyperplane} belongs to both a separating hyperplane $H_{kl}$ and the interior of $\mathcal{U}_k$, thus it does not belong to a previously discovered hyperplanes defining $\mathcal{U}_k$.
	As a result, Algorithm~\ref{alg:hyperplane} returns a new hyperplane with probability at least $1- 2\delta (m+n)^2$ thanks to Lemma~\ref{lem:find_hyperplane}. Furthermore, by observing that Algorithm~\ref{alg:hyperplane} is invoked at most $n$ different times and each time it returns a new hyperplane with a probability of $1-2\delta (m+n)^2$, we have that the while loop at Line~\ref{line:partition_loop3} terminates in a finite number of steps with probability at least $1-2\delta n (m+n)^2$, by employing a union bound over all the possible calls to Algorithm~\ref{alg:hyperplane}.
	%
	
	Moreover, we show that, if during the execution of the while loop at Line~\ref{line:partition_loop3} Algorithm~\ref{alg:hyperplane} correctly identifies a new separating hyperplane for the region $\mathcal{P}_k$ whenever invoked, then $\mathcal{U}_k = \preg_k$. This is because, throughout the execution of the while loop at Line~\ref{line:partition_loop3}, it always holds that $\mathcal{P}_k \subseteq \mathcal{U}_k$, by how Algorithm~\ref{alg:learning_commitment} works. At the same time, when such a loop terminates, we have $\mathcal{U}_k= \text{co}(V(\mathcal{U}_k))$ with $a_f^\star(p)=a_k$ for each $p \in V(\mathcal{U}_k)$, and thus $\mathcal{U}_k \subseteq \mathcal{P}_k$, proving that $\mathcal{P}_k=\mathcal{U}_k$.
	
	As a result, after the execution of the while loop at Line~\ref{line:partition_loop3}, by combining the above observations, we have that the set $\mathcal{C}$ is updated as $\mathcal{C} \gets \mathcal{C} \cup \{a_k\}$, where $\textnormal{vol}(\mathcal{P}_{a_k}) > 0$, $\mathcal{U}_k = \preg_k$, and $a_k \not \in \mathcal{C}$ with probability at least $1- \delta n (2(m+n)^2 + n )$. Moreover, we observe that the initial step of the induction, \emph{i.e.}, when $\mathcal{C}=\varnothing$, can be proved with the same argument as above. Therefore, we conclude that when Algorithm~\ref{alg:learning_commitment} terminates, the union of the sets $\mathcal{U}_j$ with $a_j \in \mathcal{C}$ coincides with the simplex $\Delta_m$ with probability at least $1-\delta n^2(2(m+n)^2 + n)$. Such a result follows by induction and employing a union bound over all the possible executions of the while loop at Line~\ref{line:partition_loop3}, which are at most $n$. Therefore, by means of Lemma~\ref{lem:vertex_optimal}, we have that Algorithm~\ref{alg:learning_commitment} returns an optimal strategy to commit to in a finite number of steps wit probability at least $1-\delta n^2(2(m+n)^2 + n)$.
	
	Finally, we observe that, with probability at least $1- \delta n^2 (2(m+n)^2 + n)$, the number of samples required by Algorithm~\ref{alg:learning_commitment} is equal to ${\mathcal{O}}\left(n^2\left(m^7L\log\left(\frac{1}{\delta}\right)+ \binom{m+n}{m} \right)\right)$. This follows from the observation that during the execution of the loop at Line~\ref{line:partition_loop3}, Algorithm~\ref{alg:hyperplane} is invoked at most $n$ times, and the number of samples required by such a procedure is ${\mathcal{O}}\left(m^7L\log\left(\frac{1}{\delta}\right) \right)$, as guaranteed by Lemma~\ref{lem:find_hyperplane}. Additionally, before computing an hyperplane, in the worst-case scenario Algorithm~\ref{alg:learning_commitment} has queried every vertex of the region $\mathcal{U}_k$. Since the number of vertexes of each $\mathcal{U}_k$ is bounded by $\binom{m+n}{m}$, and such a check is performed at most $n$ times during the execution of the loop, the number of samples required to execute the while loop at Line~\ref{line:partition_loop3} is ${\mathcal{O}}\left(n\left(m^7L\log\left(\frac{1}{\delta}\right)+ \binom{m+n}{m} \right)\right)$. Consequently, observing that the number of follower's actions is $n$, the while loop at Line~\ref{line:partition_loop3} is executed at most $n$ times. Thus, the overall number of samples required by Algorithm~\ref{alg:learning_commitment} is of the order of ${\mathcal{O}}\left(n^2\left(m^7L\log\left(\frac{1}{\delta}\right)+ \binom{m+n}{m} \right)\right)$. As a result, since $\zeta = \delta n^2 (2(m+n)^2 + n)$, with probability at least $1-\zeta$, the number of samples required by Algorithm~\ref{alg:learning_commitment} is $\widetilde{\mathcal{O}}\left(n^2\left(m^7L\log\left(\nicefrac{1}{\zeta}\right)+ \binom{m+n}{m} \right)\right)$, which concludes the proof.
\end{proof}
\subsection{Omitted Proofs from \texttt{Separating-Hyperplane}}
%
%
%
%

\vectorDifferent*

\begin{proof}
	First, we observe that $p^\circ \in \textnormal{int}(\Delta_m)$. This is because the binary search is always performed between an interior point $p \in \textnormal{int}(\Delta_m)$, in which $a^\star_f(p)= a_j$, and another point that is either an interior point or a vertex $v$ in which follower's best response does not coincide with $a_j$ (as guaranteed by Algorithm~\ref{alg:learning_commitment} and Lemma~\ref{lem:vertex_two}). Thus, the resulting $p^\circ$ cannot coincide with the vertex $v$ itself. Next, we show that each $p^{+i}$ with $i \in [m]$ belongs to $\textnormal{int}(\Delta_m)$. To prove that, notice that:
	\begin{equation*}
		\sum_{i \in [m]} p^{+i}_i = \sum_{i \in [m]} p^\circ_i + \sum_{i \in [m]} \alpha (p_i-p^\circ_i) = 1
	\end{equation*}
	and:
	\begin{equation*}
		p^{+i}_j = p^\circ_j + \alpha (p_j-p^\circ_j) \ge \frac{1}{2^{B}}-  \frac{1}{m 2^{m(B+4L)-1} }  > 0,
	\end{equation*}	
	for every $j \in [m]$. Notice that the inequality above holds because $p_j^\circ$ is at least $2^{-B}$ for every $j \in [m]$, given that the bit-complexity of $p^\circ$ is bounded by $B$, and the fact that $p^\circ$ belongs to the interior of $\Delta_m$. Therefore, noticing that $\alpha (p_j-p^\circ_j) \ge -\alpha$ for every $j \in [m]$, it follows that each $p^{+i}$ belongs to $\Delta_{m}$, with each component strictly greater than zero. Thus, this proves that all the points $p^{+i}$ belong to $\textnormal{int}(\Delta_{m})$.

	As a further step, we prove that, with probability at least $1-\delta m$, all the points $p^{+i}$ with $i \in [m]$ computed by Algorithm~\ref{alg:hyperplane} are linearly independent.
	To do that, at each iteration $k \in [m]$ of the for loop at Line~\ref{alg:hyperplane_for}, we define ${H}'_{k}:=\textnormal{span}\left\{ p^{+1}, \dots ,p^{+{k}} \right\}$ as the linear space generated by the linear combinations of the points $p^{+k}$ computed up to round $k \in [m]$.
	To ensure that all these commitment are linearly independent, we have to guarantee that $p^{+(k+1)} \not \in {H}'_{k}$ for every $k \in [m-1]$. 
	This condition is true only if the point $p$ sampled from the facet $H_{k+1} \cap \Delta_m$ does not belong to the linear space defined as $${H}''_{k} \coloneqq \left\{ \sum_{l=1}^{k} \lambda_l p^{+l} + \left(1-\nicefrac{1}{\alpha} \right)p^\circ \mid \lambda_l \in \mathbb{R} \quad \forall l \in [k] \right\}.$$ This is because, with a simple calculation, we can show that $p^{+(k+1)}$ belongs to ${H}'_{k}$ if and only if $p$ sampled from $H_{k+1} \cap \Delta_m$ belongs to ${H}''_{k}$. As a result, we have to bound the probability that $p \not \in {H}''_{k} \cap (H_{k+1} \cap \Delta_m)$ for every $k \in [m-1]$. We notice that, when $k \le m-2$, the linear space ${H}''_{k}$ has dimension at most $m-2$ and does not coincide with $H_{k+1} \cap \Delta_m$, thus ${H}''_{k} \cap (H_{k+1} \cap \Delta_m)$ has dimension at most $m-3$. When $k=m-1$, we observe that $p$ sampled from $\textnormal{int}(H_1 \cap \Delta_m)$ is such that $p_i>0$ for every $i \ge 2$, and thus $p \not \in H_m$. As a result, ${H}''_{m-1} \cap \Delta_m \not \subseteq H_{m}$. Consequently, ${H}''_{k} \cap (H_{k+1} \cap \Delta_m)$ has a dimension of at most $m-3$. Thus, it holds that $H_{k+1} \cap \Delta_m \not \subseteq {H}''_{k}$ for every $k \in [m-1]$. Consequently, Algorithm~\ref{alg:sample_point} ensures that the probability of sampling $p \in H''_{k} \cap \left(H_{k+1} \cap \Delta_m\right)$ is at most $\delta$, for every $k \in [m-1]$. Therefore, by employing a union bound, all the points $p^{+k}$ with $k \in [m]$ are linearly independent with probability at least $1-\delta m$
	
	
	To conclude the proof, we bound the probability that each $p^{+i}$ with $i \in [m]$ does not belong to the separating hyperplane $H_{jk}$. We observe that $p^{+i}$ belongs to $H_{jk} \cap \Delta_m$ only when the point $p$ sampled from $H_i \cap \Delta_m$ belongs to $H_{jk}$. This is because $p^{+i}$ is on the line segment connecting $p$ to $p^\circ$, and such a segment belongs to $H_{jk}$ only if $p$ also belongs to $H_{jk}$. Additionally, we notice that, for every $i \in [m]$, the linear space $H_{jk} \cap \left(H_i \cap \Delta_m\right)$ has dimension $m-2$ only in the case in which $H_i = H_{ij}$. Such a scenario is not possible since $p^\circ \in \textnormal{int}(\Delta_m) \cap H_{ij}$, and thus $p^\circ_i \neq 0$, which in turn implies that $p^\circ \not \in H_i$ for every $i \in [m]$. Therefore, the dimension of $H_{jk} \cap \left(H_i \cap \Delta_m\right)$ is at most $m-3$. As a result, Algorithm~\ref{alg:sample_point} ensures that the probability of sampling $p \in H_{jk} \cap \left(H_i \cap \Delta_m\right)$ is at most $\delta$, since $\left( H_i \cap \Delta_m \right) \not \subseteq H_{jk}$ for every $i \in [m]$. As a result, by employing a union bound over all the points sampled from the $m$ facets of the simplex, we have that $p^{+i} \not \in $ $H_{jk}$ for all $i \in [m]$ with probability at least $1-\delta m$. 
	
	Finally, by employing a union bound, we have that the lemma holds with probability at least $1-2\delta m$.
\end{proof}

\intersectionNull*

\begin{proof}
	As a first step, we bound the probability that the point $p^\circ$ computed at Line~\ref{line:binary_search_hp_1} in Algorithm~\ref{alg:hyperplane} belongs to a \emph{single} hyperplane $H_{jk}$ with $k \in [n]$ or to multiple coinciding hyperplanes $H_{jk}$ for $k \in [n]$, and thus, $p^\circ$ belongs to the interior of a facet of $\mathcal{P}_j$.
	Let $H, H’$ be either separating hyperplanes or boundary planes defining two different facets of $\mathcal{P}_{j} \subseteq \Delta_m$.
	In the following, we bound the probability that the line segment connecting $p^1$ to $p^2$ does not intersect any $H \cap H' $ with non-empty intersection.
	To prove that, we define $\mathcal{P} = H \cap H'$ as the linear space defined as the non-empty intersection of $H$ and $H'$. We observe that such a linear space has a dimension of $m-2$, given that $H$ and $H'$ are distinct and non-parallel hyperplanes corresponding to different facets of $\mathcal{P}_{j}$. 
	Consequently, there exists a linear space $H''$, defined as the linear combination of the point $p^2$ where the binary search is conducted and the linear space $\mathcal{P}$, formally defined as $H'' = \textnormal{span} \{ p^2, \mathcal{P} \}$, with a dimension at most $m-1$. 
	To prove the statement of the lemma, we first consider the case in which the randomly sampled point $p^1$ is such that $a^\star_f(p^1)=a_j$. Then, in order to ensure that the binary search determines a point $p^\circ$ that belongs to a single or possibly coinciding hyperplanes $H_{jk}$, we have to bound the probability that $p^1 \not \in H''$. We  notice that such a probability is greater or equal to  $1-\delta (m+n)^2$.
	Such a result follows by means of both Lemma~\ref{lem:sample_point_first} and a union over all the possible pairs of hyperplanes $H$ and $H'$ defining the facets of $\mathcal{P}_j$, which are at most $\binom{m+n}{2} \le (m+n)^2$.
	Furthermore, if $a^\star_f(p^1) \neq a_j$ and $p^2=p^{\textnormal{int}}$, we can employ the same argument to bound the probability that $p^1$ does not belong to $H''$.
	Thus, by combining the two previous results and employing the law of total probability,  the point $p^\circ$ lies on a single or possibly coinciding separating hyperplanes $H_{jk}$ with $k \in [n]$ with probability of at least $1-\delta (m+n)^2$.
	
	
	We also define $\mathbb{S}^{m}_{\epsilon}(p^\circ)\coloneqq \left\{ x \in \mathbb{R}^m \,|\, \| x- p^\circ\|_2 \le \epsilon \right\}$ as the sphere in $\mathbb{R}^m$ of radius $\epsilon$ centered in $p^\circ$ and we show that when $\epsilon<\epsilon'=  2^{-m(B+4L)}/m$ such a sphere has null intersection with the boundaries of both $\mathcal{P}_j$ and $\mathcal{P}_k$.
	To do that, we prove that the $\| \cdot \|_2$-distance between the commitment $p^{\circ}$ and any separating hyperplane defining the facets of both $\mathcal{P}_j$ and $\mathcal{P}_k$ can be computed as follows:
	\begin{equation*}
		d_2(p^{\circ},\widetilde H_{lm}) = \left| \frac{ \displaystyle\sum_{i=1}^{m} p^{\circ}_i \Big(u^f(a_i,a_{l}) -u^f(a_i,a_{m})\Big)}{ \sqrt{\displaystyle\sum_{i=1}^{m} \Big( u^f(a_i,a_{l}) -u^f(a_i,a_{m}) \Big)^2} } \right| \ge \frac{1}{m 2^{m(B+4L)}}
	\end{equation*}
	for every $(l,m) \in \bigcup_{k' \in [n]}(k',j) \cup \bigcup_{ j' \in [n]}(k,j')$ such that $H_{lm}$ does not coincide with $H_{jk}$. We observe that the inequality follows by means of Lemma~\ref{lem:sum_rationals}, as the numerator of the above fraction is the sum of $m$ terms whose bit complexities are bounded by $B+4L$. Furthermore, the denominator of such a quantity is at most $m$. Analogously, we can prove that the same lower bound holds for each boundary plane $H_i$ with $i \in [m]$, defining the facets of either $\mathcal{P}_j$ or $\mathcal{P}_k$. Consequently, since $p^{\circ}$ belongs to a single hyperplane $H_{jk}$ with $\textnormal{vol}(\mathcal{P}_{k})>0$, we have that the sphere  $\mathbb{S}^{m}_{\epsilon}(p^\circ) \cap \Delta_m$ belongs to $\textnormal{int}(\mathcal{P}_j \cup \mathcal{P}_k)$.
	
	We also show that all the commitments $p^{+i}$ with $i \in [m]$ computed in Algorithms~\ref{alg:hyperplane} are such that $p^{+i} \in \mathbb{S}^{m}_{\epsilon}(p^{\circ}) $ with $\epsilon < \epsilon' =~ 2^{-m(B+2L)}/m$. Formally, we have:
	\begin{equation*}
		\| p^{+i} - p^{\circ}\|_2 = \| p^{\circ} + \alpha (p - p^{\circ}) - p^{\circ}\|_2 = \alpha \| p^{\circ}- p \|_2 \le 2^{-m(B+4L)-1/2} / m < \epsilon'.
	\end{equation*}
	The first equality holds because of the definition of $p^{+i}$, while the second inequality holds because $\max_{p,p' \in \Delta_m} \| p-p'\|_2 \le \sqrt{2}$ and thanks to the definition of $\alpha =  2^{-m(B+4L)-1}/m$. Thus, observing that each $p^{+i} \in \mathbb{S}^{m}_{\epsilon}$ for each $i \in [m]$, we have that in each $p^{+i}$ either $a_j$ or $a_k$ is a best-response.
	
	Finally, we show that if $a^\star_f(p^{+i})=a_j$, then $a^\star_f(p^{-i})=a_k$. More specifically, if $a^\star_f(p^{+i})=a_j$ and $p^{+i} \not \in H_{jk} $, then the following holds:
	\begin{equation*}
		\sum_{i=1}^{m} p^{+i}_i \Big(u_f(a_i,a_j) - u_f(a_i,a_k) \Big) = \alpha \sum_{i=1}^{m} (p^\circ_i - p_i) \Big(u_f(a_i,a_j) - u_f(a_i,a_k) \Big) > 0.
	\end{equation*}
	This, in turn, implies that:
	\begin{equation*}
		0 > -\alpha \sum_{i=1}^{m} \Big(p^\circ_i - p_i) (u_f(a_i,a_j) - u_f(a_i,a_k) \Big) = \sum_{i=1}^{m} p^{-i}_i \Big(u_f(a_i,a_j) - u_f(a_i,a_k) \Big),
	\end{equation*}
	since $p^\circ \in H_{jk}$. Finally, with the same argument, we can show that if $a^\star_f(p^{+i})=a_k$, then $a^\star_f(p^{-i})=a_j$, concluding the proof.
\end{proof}

\hyperSegmen*
\begin{proof}

To prove the lemma we observe that all the points $p^{+i}$ with $i \in [m]$ are linearly independent and do not belong to the separating hyperplane $H_{jk}$ with probability $1-2\delta m$, by means of Lemma~\ref{lem:lin_indep}. Furthermore, employing Lemma~\ref{lem:hyper_ball}, we have that in each point $p^{+i}$ either $a_j$ or $a_k$ is a follower's best response. As a result, employing a union bound, we can compute by binary search $m-1$ points lying on the separating hyperplane $H_{jk}$ with probability at least $1-2 \delta(m+n)^2$.

To conclude the proof, we observe that the binary search to compute $p^\circ$ is consistently performed between a point sampled according to Algorithm~\ref{alg:sample_point}, whose bit-complexity is bounded by ${\mathcal{O}}(Lm^3 + \log(\frac{1}{\delta}))$ as a consequence of Lemma~\ref{lem:sample_point_first}, and either another interior point or a vertex. Therefore, since the bit-complexity of each vertex can be bounded by $O(m^2L)$ by means of Lemma~\ref{lem:vertex_bits}, the number of rounds required to compute $p^\circ$ by means of Algorithm~\ref{alg:binary_search} is ${\mathcal{O}}(m^4L + m\log(\frac{1}{\delta}))$ as prescribed by Lemma~\ref{lem:binary_search}. Furthermore, the bit-complexity of $p^\circ$ is bounded by ${\mathcal{O}}(m^4L + m\log(\frac{1}{\delta}))$. As a result, the bit-complexity of the parameter $\alpha$ is equal to $\mathcal{O}(m^5L + m^2\log(\frac{1}{\delta}))$. Thus, the bit-complexity of each point $p^{+i} = p^\circ + \alpha (p - p^\circ)$ with $p \in H_i \cap \Delta_m$ is bounded by $\mathcal{O}(m^5L + m^2\log(\frac{1}{\delta}))$. Consequently, each of the final $m-1$ binary searches performed to compute the $m-1$ points on the separating hyperplane requires at most $\mathcal{O}(m^6L + m^3\log(\frac{1}{\delta}))$ samples and returns a points on the hyperplane with bit-complexity bounded by $\mathcal{O}(m^6L + m^3\log(\frac{1}{\delta}))$. Therefore, Algorithm~\ref{alg:hyperplane} requires $\mathcal{O}(m^7L + m^4\log(\frac{1}{\delta}))$ samples by accounting for all the $m-1$ binary search performed.
\end{proof}


\subsection{Proofs Omitted from \texttt{Binary-Search}}
\binarySearch*
\begin{proof}
	We notice that the point $p^\circ$ returned by Algorithm~\ref{alg:binary_search} lies on line segment that connect $p^1$ to $p^2$. This implies that $p^\circ \gets p^1 + \lambda^\circ (p^2 - p^1)$ for some $\lambda^\circ \in [0,1]$. Furthermore, $p^\circ$  also lies on a separating hyperplane $H_{jk}$, thus the following holds $\sum_{i \in [m]}p^\circ _i (u_f(a_i,a_j) - u_f(a_i,a_k)) = 0$. Then, by combining the two equations we get:
	\begin{equation*}
		\lambda^\circ =  \frac{ \left| \displaystyle\sum_{i=1}^{m}p^1_i \Big(u_f(a_i,a_j) -u_f(a_i,a_k)\Big) \right| }{ \left | \displaystyle\sum_{i=1}^{m} (p^2_i - p^1_i)\Big(u_f(a_i,a_j) -u_f(a_i,a_k)\Big) \right | } .
	\end{equation*}

	We observe that the bit complexity of the numerator defining $\lambda^\circ$ is the sum of $m$ terms with a bit complexity bounded by $B +4L$. This is because the terms encoding the difference in the follower's utility, as prescribed by Lemma~\ref{lem:sum_rationals}, have a bit complexity bounded by $4L$. Therefore, the sum of these $m$ terms has a bit complexity of $3m(B+4L)$ due to Lemma~\ref{lem:sum_rationals}. Similarly, we can show that the denominator that defines $\lambda^\circ$ has a bit complexity bounded by $3m(4B+4L)$. As a result, the bit complexity of $\lambda^\circ$ is bounded by $3m(5B+8L)$, which is of the order $\mathcal{O}(m(B+L))$. 
	
	Moreover, the minimum distance between two valid values of $\lambda^\circ$ is lower-bounded by $ \epsilon\coloneqq2^{-6m(5B+8L)}$,  according to Lemma~\ref{lem:sum_rationals}. Therefore, when Algorithm~\ref{alg:binary_search} terminates, there exists a unique $\lambda^\circ$ within the interval $[\lambda_1, \lambda_2]$, with bit complexity bounded by $3m(5B+8L)$, given that the distance between $\lambda_1$ and $\lambda_2$ satisfies $|\lambda_1 - \lambda_2| \le \epsilon $. Consequently, the bit complexity of the resulting point $p^\circ$ lying on the separating hyperplane $H_{jk}$ is bounded by $\mathcal{O}(m(B + L))$. This is because $p^\circ \gets p^1 + \lambda^\circ (p^2 - p^1)$ is defined as the sum of two rational numbers with bit complexities bounded by $3m(5B+8L)+B$, and thus it is bounded by $24m(3B+4L)$. Lastly, we observe that the number of rounds required by Algorithm~\ref{alg:binary_search} to terminate is on the order of $\mathcal{O}(\log_2(\nicefrac{1}{\epsilon}))$, and thus $\mathcal{O}(m(B+L))$.
\end{proof}

%
%
%

\subsection{Proofs Omitted from \texttt{Sample-Point}}
\samplepointfirst*
\begin{proof}
	
	In the following, for the sake of the presentation, we provide the proof of the lemma when the polytope $\mathcal{P}$ is such that $\textnormal{vol}_{m-1}(\mathcal{P})>0$, the case when $\mathcal{P}=H_i \cap \Delta_m$ can be proved with the same analysis. In the following, we let $H~\coloneqq \{x \in \mathbb{R}^{m} \,|\, \sum_{i=1}^m \alpha_i x_i =~\beta \}$, be an hyperplane in $\mathbb{R}^{m}$. We show that the probability that the commitment $p$ returned by Algorithm~\ref{alg:sample_point} belongs to $H$ is equal to the probability that $x \in \Xi \coloneqq \{- 1, -\frac{M-1}{M}, \dots, 0, \dots, \frac{M-1}{M}, 1 \}^{m-1}$, sampled according to a uniform probability distribution over the set $\Xi$, belongs to a suitably defined linear space $H'$.
	\begin{align*}
		\mathbb{P}(p \in H) & = \mathbb{P}\left( \sum_{i=1}^{m-1} \alpha_i  p_i+ \alpha_m \left( 1 - \sum_{i=1}^{m-1} p_i \right) = \beta \right)\\
		& = \mathbb{P}\left( \sum_{i=1}^{m-1} \left(\alpha_i -\alpha_m\right)  p_i  = \beta - \alpha_m \right)\\
		& = \mathbb{P}\left( \sum_{i=1}^{m-1} \left(\alpha_i -\alpha_m\right) \left( p^\diamond_i + \rho x_i \right) = \beta - \alpha_m \right)\\
		& = \mathbb{P}\left( \rho \sum_{i=1}^{m-1} \left(\alpha_i -\alpha_m\right) x_i = \beta - \alpha_m - \sum_{i=1}^{m-1} \left(\alpha_i -\alpha_m\right)  p^\diamond_i   \right)\\
		& = \mathbb{P}\left( \sum_{i=1}^{m-1} \left(\alpha_i -\alpha_m\right) x_i = \frac{1}{\rho}	\left(\beta - \alpha_m - \sum_{i=1}^{m-1} \left(\alpha_i -\alpha_m\right)  p^\diamond_i   \right)\right)\\
		& = \mathbb{P}\left( x \in H'\right),
	\end{align*}
	where we let $H'=\left \{ x \in \mathbb{R}^{m-1} \,|\, \sum_{i=1}^{m-1} \left(\alpha_i -\alpha_m\right) x_i = \frac{1}{\rho}	\left(\beta - \alpha_m - \sum_{i=1}^{m-1} \left(\alpha_i -\alpha_m\right)  p^\diamond_i   \right) \right \}$.
	Furthermore, we observe that a discrete probability over the set $\Xi$ can be written as follows:
	\begin{align*}
		\mathbb{P}_{x \sim \mathcal{U}(\Xi)}\left( x = \bar x \right) & = \frac{1}{\textnormal{vol}(Q)} \int_{Q} \mathbbm{1} \left \{  \| s - \bar x\|_\infty \le \frac{1}{2M} \right \} ds
	\end{align*}
	for each $\bar x \in \Xi$, where $Q \coloneqq \{ x \in \mathbb{R}^{m-1} \,| \, |x_i| \le 1+1/2M \,\, \forall i \in [m-1]\}$ is an hypercube in dimension $m-1$. Then, the following holds: 
	\begin{align*}
		\mathbb{P}_{x \sim \mathcal{U}(\Xi)}\left( x \in H' \right) & = \sum_{\bar x \, \in \, \Xi } \mathbbm{1}\left \{ \bar x \in H' \right \} \mathbb{P}_{x \sim \mathcal{U}(\Xi)}\left( x = \bar x \right) 	\\
		& = \frac{1}{\textnormal{vol}(Q)}  \sum_{\bar x \, \in \, \Xi } \mathbbm{1}\left \{ \bar x \in H' \right \}  \int_{Q} \mathbbm{1} \left \{  \| s - \bar x\|_\infty \le \frac{1}{2M} \right \} ds	\\
		& = \frac{1}{\textnormal{vol}(Q)}  \sum_{\bar x \, \in \, \Xi }  \int_{Q} \mathbbm{1} \left \{ s \in Q  \,\, \Big| \,\, \| s - \bar x\|_\infty \le \frac{1}{2M}, \,  \bar x \in H' \right \} ds\\
		& =   \frac{1}{\textnormal{vol}(Q)} \int_{Q} \mathbbm{1} \left \{ s \in Q  \,\, \Big| \,\, \| s - \bar x\|_\infty \le \frac{1}{2M}, \, \bar x \in H' \cap \Xi  \right \} ds\\
		& \le   \frac{1}{\textnormal{vol}(Q)} \int_{Q} \mathbbm{1} \left \{  s \in Q  \,\, \Big| \,\,  \| s - t\|_\infty \le \frac{1}{2M}, t \in H' \right \} ds\\
		& = \frac{1}{\textnormal{vol}(Q)} \int_{Q} \mathbbm{1} \left \{  d_\infty(s,H') \le \frac{1}{2M} \right \} ds\\
		& \le \frac{1}{\textnormal{vol}(Q)} \int_{Q} \mathbbm{1} \left \{  d_2(s,H') \le \frac{\sqrt{m}}{2M} \right \} ds\\
	\end{align*}
	Thus, we can bound the probability that the commitment sampled according to Algorithm~\ref{alg:sample_point} belongs to the separating hyperplane $H$ as follows:
	\begin{align*}
		\mathbb{P}(p \in H)  = \mathbb{P}_{x \sim \mathcal{U}(\Xi)}\left( x \in H' \right) &\leq 2 \frac{ \displaystyle\int_{Q \cap H'}  \frac{\sqrt{m} }{2M} dA }{\textnormal{vol}_{m-1}(Q )}\\
		&= \frac{\sqrt{m} }{M }\frac{ \displaystyle\int_{Q \cap H'} dA }{ (2+1/M)^{m-1}} \\
		&\leq \frac{ \sqrt{2m}}{M }\frac{(2+1/M)^{m-2}}{ (2+1/M)^{m-1} }  = \frac{\sqrt{m}}{\sqrt{2}M} \le \delta. 
	\end{align*}
	Where the second inequality above holds by employing the results by~\citet{Ball1986}~and~\citet{Ivano2021} and the last inequality follows since $M =  \lceil {\sqrt{m}}/{\delta} \rceil$, showing that $\mathbb{P}(p \in H) \leq \delta$. Furthermore, we notice that for each linear space $\widetilde{H}$ with a dimension less than $m-1$, there always exists a hyperplane $H$ of dimension $m-1$ such that $\widetilde{H} \subseteq H$ and thus:
	$$\mathbb{P}(p \in \widetilde H) \leq\mathbb{P}(p \in H) \leq \delta.$$
	
	
	In the following, we show that the point $p$ sampled according to Algorithm~\ref{alg:sample_point} belongs to the polytope $\mathcal{P}$. To do that, we denote the $h$-vertex of the set $\mathcal{V}$ as $v^h$, so that $\mathcal{V} = \{v^1,\dots,v^m\}$, with $V= |\mathcal{V}|$.
	Furthermore, given that $p^\diamond \in \text{int}(\mathcal{P})$, the $\| \cdot \|_2$-distance between $p^\diamond$ and any separating hyperplane $H_{jk}$ that defines the boundary of $\mathcal{P}$ can be lower-bounded as follows:
	\begin{align*}
		d_2(p^\diamond, H_{jk}) & = \left| \frac{ \displaystyle\sum_{i=1}^{m} p^\diamond_i\Big(u_f(a_i,a_j) -u_f(a_i,a_k)\Big)}{ \sqrt{\displaystyle\sum_{i=1}^{m} \Big( u_f(a_i,a_j) -u_f(a_i,a_k) \Big)^2  } } \right| \\
		 & = \left|  \frac{ \displaystyle\sum_{i=1}^{m} \sum_{h=1}^{m} \Big(u_f(a_i,a_j) -u_f(a_i,a_k)\Big)  v^h_i}{ m \sqrt{\displaystyle\sum_{l=1}^{m} \Big( u_f(a_i,a_j) -u_f(a_i,a_k) \Big)^2}   } \right| \\
		& \geq \left(m^2 2^{9m^3L +2mL} \right)^{-1} \eqqcolon r.
	\end{align*}
	To prove the last inequality, we provide a lower bound for the numerator and an upper bound for the denominator of the fraction above. More specifically, the denominator can be upper bounded by $m^2$ since $|u_f(a_i, a_j) - u_f(a_i, a_k) | \le 1$ for each $i \in [m]$.
	To lower bound the numerator, we define: 
	$$\frac{\beta_i}{\gamma_i} := u_f(a_i,a_j) -u_f(a_i,a_k) \quad \forall i \in [m] \quad\text{and}\quad v_{i}^h := \frac{\mu_i^{h}}{\nu_h} \quad \forall v^h \in \mathcal{V},$$
	where we let ${\beta_i}$ and ${\gamma_i}$ be integer numbers for each $i \in [m]$, while ${\nu_h}$ and ${\mu^{h}_i}$ are natural numbers for each $i \in [m]$ and $h \in [V]$. As a result, we have:
	\begin{align*}
		\left|  \sum_{i=1}^{m} \sum_{h=1}^{m} \Big(u_f(a_i,a_j) -u_f(a_i,a_k) \Big)  v^h_i \right| 
		&= \left|  \sum_{i=1}^{m} \sum_{h=1}^{m} \frac{\beta_i \mu^{h}_{i}}{\gamma_i \nu_h} \right| \\
		&= \left| \frac{\displaystyle\sum_{i=1}^{m} \sum_{h=1}^{m} \beta_i \mu^{h}_{i} \left( \prod_{i' \neq i} \gamma_{i'} \prod_{h' \neq h} \nu_{h'} \right)}{\displaystyle\prod_{i=1}^{m}\gamma_i \prod_{h=1}^{m}\nu_h} \right| \\
		&\geq \left( \prod_{i=1}^{m}\gamma_i \prod_{h=1}^{m}\nu_h \right)^{-1}
		\geq 2^{-(9m^3L + 4mL)}.
	\end{align*}
	We observe that the first inequality holds because the numerator defining the above fraction can be lower bounded by one. On the other hand, the denominator can be upper bounded by noting that the bit complexity of $\gamma_i$ is bounded by $4L$ for each $i \in [m]$, as ensured by Lemma~\ref{lem:sum_rationals}. Additionally, for every $h \in [m]$, the bit complexity of $\nu_h$ is bounded by $9m^2L$ thanks to Lemma~\ref{lem:vertex_bits}. With a similar argument, we can show that the distance between $p^\diamond$ and any boundary hyperplane is lower bounded by $r$.
	
	As a final step, we bound the $\| \cdot \|_2$ distance between $p^\diamond$ and $p$. Formally:
	\begin{align*}
		d_2(p^\diamond,p) &= \sqrt{\sum_{i \in [m-1]}\Big(p^\diamond_i + \rho x_i -p^\diamond_i\Big)^2 + \left(p^\diamond_m - 1 + \sum_{i \in [m-1]}p^\diamond_i + \rho x_i\right)^2} \\
		&\leq \sqrt{\sum_{i \in [m-1]}\Big(p^\diamond_i + \rho x_i -p^\diamond_i\Big)^2 + \left(\sum_{i \in [m]}p^\diamond_i + \sum_{i \in [m-1]}\rho x_i - 1\right)^2} \\
		&\leq \sqrt{\sum_{i \in [m-1]}\Big(\rho x_i\Big)^2 + \left( \sum_{i \in [m-1]}\rho x_i \right)^2} \\
		&\leq \sqrt{(m-1)\rho^2 + (m-1)^2 \rho^2} = \rho \sqrt{(m-1)m} < \rho m
	\end{align*}
	As a result, if $\rho = {r}/{m}$ the distance $d_2(p^\diamond,p)$ is strictly smaller than $r$, showing that the commitment $p \in \textnormal{int}(\mathcal{P})$.
	
	 
	We also show that the bit complexity of $p$ is bounded by $40m^3L + 2\log_2(\nicefrac{1}{\delta})$. We observe that the denominator of each rational number $x_i$ with $i \in [m-1]$ is equal to $M =  \lceil {\sqrt{m}}/{\delta} \rceil$. Furthermore, each component of $p^\diamond$ can be expressed as follows:
	$$p^\diamond_i = \frac{1}{m}\sum_{h=1}^{m}v^h_i = \frac{1}{m}\sum_{h=1}^{m}\frac{\mu^{h}_{i}}{\nu_h}= \frac{\displaystyle\sum_{i=1}^m \mu^h_i \left( \prod_{h \ne h'} \nu_{h'} \right)}{m\displaystyle\prod_{h=1}^{m}\nu_h},$$
 	for each $i \in [m]$. 
 	Consequently, the denominator of each rational component $p_i = p_i^\diamond + \rho x_i$ with $i \in [m-1]$ is equal to $D \coloneqq m\prod_{h=1}^{m}\nu_h D_{\rho} M$, where $D_{\rho}$ represents the denominator defining the rational number $\rho$. Additionally, we observe that the last component $p_m$ 
 	can also be expressed as a rational number with the same denominator $D$.
 	
 	Finally, since $p_i \in [0,1]$ for every $i \in [m]$, its bit complexity is bounded by $2B_D$, where the bit complexity of the denominator $D$ can be upper bounded as follows:
	\begin{align*}
		B_D &= \left \lceil \log_2\left(m\prod_{h=1}^{m}\nu_h D_{\rho} M\right) \right \rceil \\
		&= \left \lceil \log_2(m) + \sum_{h=1}^{m}\log_2(\nu_h) + \log_2(m 2^{9m^3L +4mL}) + \log_2(M)\right \rceil \\
		&\le 1 +  \log_2(m) + \sum_{h=1}^{d}\log_2(2^{9m^2L}) + \log_2(m 2^{9m^3L +4mL}) + \log_2(\nicefrac{\sqrt{m}}{\delta} +1) \\
		&\le 20m^3L + \log_2(\nicefrac{1}{\delta}).
	\end{align*}
	Thus, the bit complexity of $p$ is bounded by $40m^3L + 2\log(\nicefrac{1}{\delta})$, concluding the proof.
\end{proof}

\section{{Additional Technical Lemmas}}\label{sec:omitted_lemma_proof}

\begin{restatable}{lemma}{sumrationals}
	\label{lem:sum_rationals}
	Given $q_1, \dots, q_m \in \mathbb{Q}$ represented as fractions, each with bit-complexity bounded by $B \in \mathbb{N}_{+}$, the bit-complexity of their sum $q \coloneqq \sum_{i=1}^m q_i$ is bounded by $4B$ when $m=2$, while it is bounded by $3mB$ when $m>2$. Moreover, the absolute value of their sum is either equal to $0$ or it is greater than or equal to $2^{-Bm}$.
\end{restatable}
\begin{proof}
To prove the lemma, we define the sum of $q_1, \dots, q_m$ as $q \in \mathbb{Q}$, and we express each $q_i$ with $i \in [m]$ as $\alpha_i/\beta_i$, where both $\alpha_i,\beta_i \in \mathbb{Z}$ for each $i \in [m]$. Then the following holds:
	\begin{equation*}
	q = \sum_{i=1}^m q_i = \sum_{i=1}^m \frac{\alpha_i}{\beta_i} =  \frac{\displaystyle\sum_{i=1}^m \alpha_i \left( \prod_{i \ne j} \beta_j \right)  }{\displaystyle\prod_{i \in [m]}\beta_i}.
	\end{equation*}
It is easy to see that the bit complexity of the denominator that defines the rational number $q$ is equal to $\sum_{i \in[m]} B_{\beta_i}$, while the absolute value of the numerator of $q$ can be upper-bounded as follows:
\begin{equation*}
\left| \sum_{i=1}^m \alpha_i \prod_{i \ne j} \beta_j \right| \le m  \Bigg|  \left(\alpha_{i^*} \prod_{i \ne i^* } \beta_i \right)  \Bigg|  \le 2^{\log_2(m) + B_{\alpha_{i^*}} + \sum_{i\ne i^*} B_{\beta_{i}}},
\end{equation*}
where we let $i^* \in \argmax_{i \in [m]} \alpha_i \prod_{i \ne j} \beta_j$. As a result, the number of bits required to encode $q$ can be bounded as:
\begin{align*}
B_q &\le \log_2(m) + B_{\alpha_{i^*}} + \sum_{i \ne i^*} B_{\beta_i} + \sum_{i \in[m]} B_{\beta_i } \\
&\le \log_2(m) + B + 2 \sum_{i \ne i^*} B_{\beta_i} \\
&\le \log_2(m) + (2m-1) B \le  3mB.
\end{align*}
The second inequality above holds because $B_{\alpha_i} + B_{\beta_i} \le B$, as the bit complexity of each $q_i$ with $i \in [m]$ is bounded by $B$. Moreover, when $m=2$, we have that $B_q \le \log_2(2) + (4-1) B \le 4B$. Finally, we observe that the absolute value of $q$ is either zero or can be lower bounded as shown in the following:
\begin{equation*}
 \left|  \sum_{i=1}^m q_i \right| = \left| \sum_{i=1}^m \frac{\alpha_i}{\beta_i} \right| = \left|  \frac{\displaystyle\sum_{i=1}^m \alpha_i \left(\prod_{i \ne j} \beta_j \right) }{\displaystyle\prod_{i \in [m]}\beta_i} \right| \ge \frac{1}{2^{mB}},
\end{equation*}
which concludes the proof.
\end{proof}

\begin{restatable}{lemma}{vertexbits}
	\label{lem:vertex_bits}
	Given an SG, if follower's payoffs are represented as fractions with bit-complexity bounded by $L$, then each vertex $v \in V(\mathcal{U}_{j})$ of an $\mathcal{U}_j$ computed by Algorithm~\ref{alg:learning_commitment} has bit-complexity at most $9Lm^2$. Furthermore, with a bit-complexity of $9Lm^2$, all the components of the vector identifying a vertex can be written as fractions with the same denominator.
\end{restatable}

\begin{proof}
	Let $v$ be a vertex belonging to an upper bound $\mathcal{U}_{j}$ such that $\text{vol}(\mathcal{U}_{j})>0$. Then such a vertex lies on the hyperplane $H'$ ensuring that the sum of its components is equal to one. Furthermore, it also belongs to a subset of $m-1$ linearly independent hyperplanes, which can pertain either to the boundary planes denoted as $H_i$ with $i \in [m]$, or to the separating hyperplanes between two followers' best response regions, \emph{i.e}., $H_{jk}$ with $a_j,a_k \in \mathcal{A}_f$. Consequently, there exists a matrix $A \in \mathbb{Q}^{m \times m}$ and a vector $b \in \mathbb{Q}^{m}$ such that $Av=b$. The entries of the matrix $A$ may encode the difference in terms of follower's utility between two follower actions $a_j$ and $a_k$ in $\mathcal{A}_{f}$ for each leader's action $a_i \in \mathcal{A}_{\ell}$, \emph{i.e}., $u_f(a_i,a_j) - u_f(a_i,a_k)$. As a result, the bit complexity of each entry of the matrix $A$ can be bounded by $4L$, as a consequence of Lemma~\ref{lem:sum_rationals}, and by observing that the coefficients that defines the hyperplanes $H'$ or $H_i$, with $i \in [m]$, are either 0 or 1. 
	Therefore, we can multiply each row of the matrix $A_i \in \mathbb{Q}^{m}$ and the corresponding $b_i \in \mathbb{Q}$ with $i \in [m]$, 
	by a constant bounded by $2^{4Lm}$.
	
	In this way, we can formulate the previous linear systems to compute the vertex $v$ using a matrix $A' \in \mathbb{Z}^{m \times m}$ and a vector $b' \in \mathbb{Z}^{m}$, such that $A'v=b'$, where each component of $A'$ and $b'$ satisfies $|a'_{ij}| \le 2^{4Lm}$ and $|b'_{j}| \le 2^{4Lm}$ for each $i,j \in [m]$.	To prove the lemma, we employ a similar approach to the one presented in Lemma 8.2 by~\citet{Bertsimas}, and we define ${A'}(j)$ as the matrix obtained by substituting the $j$-th column of $A'$ with $b'$. Then, by Cramer's rule, the value of the $j$-th component of $v_j$ can be computed as follows:
	\begin{equation*}
		v_j = \frac{\det(A'(j))}{\det(A')} \,\,\,\ \textnormal{$\forall j \in [m]$}.
	\end{equation*}
	We observe that both determinants are integer numbers as the entries of both $A'$ and $b'$ are all integers, thus by Hadamard's inequality we have:
	\begin{equation*}
		|\det(A')| \leq \hspace{-1.5mm} \prod_{i \in [m]}\sqrt{\sum_{j \in [m]} {{a}'_{ji}}^2 } \leq \hspace{-1.0 mm} \prod_{i \in [m]}\sqrt{\sum_{j \in [m]} (2^{4Lm})^2 } = \hspace{-1.5mm}
		\prod_{i \in [m]}\hspace{-1mm}\sqrt{m(2^{4Lm})^2} 
		= (2^{4Lm^2}) m^{\frac{m}{2}}.
	\end{equation*}
	Furthermore, observing that the same upper bound holds for $|\det(A^j)|$, the bit complexity of $v$ is bounded by:
	\begin{equation*}
		2 \left \lceil \log_2(2^{4Lm^2} m^{m/2}) \right \rceil \le 2 ( \log_2(2^{4Lm^2}m^{m/2}) +1)=  8Lm^2 + {m} \log_2(m) +2 \le 9Lm^2.
	\end{equation*}
	Finally, we notice that this upper bound holds when every component of the vertex has the same denominator $\det(A')$, concluding the proof.
\end{proof}

\section{On the Running Time of Algorithm~\ref{alg:learning_commitment}}\label{sec:app_running_time}


\begin{restatable}{theorem}{runningtimemain}\label{thm:main_thm_running_time}
	With probability at least $1-\zeta$, the running time of Algorithm~\ref{alg:learning_commitment} is polynomial when either the number of follower's actions $n$ or that of leader's actions $m$ is fixed.  
\end{restatable}
\begin{proof}

	We observe that, by employing the same argument used to prove Theorem~\ref{thm:main_thm}, before entering the while loop at Line~\ref{line:partition_loop3}, the set $\mathcal{C}$ is either empty or is such that if $a_k \in \mathcal{C}$, then with a probability of at least $1-\zeta$, $\textnormal{vol}(\mathcal{U}_k)>0$, and $\mathcal{U}_k=\mathcal{P}_k$. Consequently, with this probability, the set $\Delta_m \setminus \bigcup_{a_k \in \mathcal{C}} \mathcal{U}_{a_k}$ is always non-empty if Algorithm~\ref{alg:learning_commitment} has not terminated.
	
	In the following, we show that for each set of actions $\mathcal{C}$ computed through the execution of Algorithm~\ref{alg:learning_commitment}, we can always efficiently sample a point from the interior of $\Delta_m \setminus \bigcup_{a_k \in \mathcal{C}} \mathcal{U}_{a_k}$ if either $n$ or $m$ is fixed. To achieve this, we define for each action $a_j \in \mathcal{C}$ the set $U(j) \subseteq {A}_f$ as $U(j) = \{a_k \in {A}_f \;|\; a_k \ne a_j \wedge H_{jk} \cap \mathcal{P}_j \text{ is a facet of } \mathcal{P}_j\} = \{a_k \in {A}_f \;|\; a_k \ne a_j \wedge \text{vol}_{m-2}(H_{jk} \cap \mathcal{P}_j) > 0 \}$.
	Furthermore, we let $P_j \coloneqq \bigcup_{a_k \in U(j)} (j,k)$ and $P \coloneqq \bigtimes_{a_j \in \mathcal{C}} P_j$. Observe that an element $\rho \in P$ indicates a pair $(j,k)$ for every $a_j \in \mathcal{C}$. 
	We denote with $\rho(j)$ the pair $(j,k)$ relative to action $a_j$. As a result, the following holds:
	\begin{align*}
		\Delta_m \setminus \bigcup_{a_k \in \mathcal{C}} \mathcal{P}_{k}  & =  \Delta_m \setminus  \bigcup_{a_j \in \mathcal{C}} \mathcal{U}_j  \\
		&=  \Delta_m \setminus \  \bigcup_{a_j \in \mathcal{C}} \Big(\mathcal{U}'_j \cap \Delta_m \Big) \\ 
		&=  \Delta_m \setminus \ \Big( \Big( \bigcup_{a_j \in \mathcal{C}} \mathcal{U}'_j \, \Big) \cap \Delta_m \Big)  \\
		&=  \Delta_m \cap \Big(  \bigcup_{a_j \in \mathcal{C}} \mathcal{U}'_j \Big)^{C}
		\\
		&= \Delta_m \cap \Big(  \bigcap_{a_j \in \mathcal{C}} {\mathcal{U}_j'}^{C} \Big)\\
		& = \Delta_m \cap \Big(  \bigcap_{a_j \in \mathcal{C}} \Big( \bigcap_{a_k \in U(j)} \mathcal{H}_{jk} \Big)^{C} \, \Big) \\
		&= \Delta_m \cap \Big(  \bigcap_{a_j \in \mathcal{C}} \Big( \bigcup_{a_k \in U(j)} \mathcal{H}^{C}_{jk} \Big)\Big)\\
		&=  \Delta_m \cap \bigcup_{\rho \, \in \, P} \, \Big(  \bigcap_{\, a_j \in \mathcal{C}} \mathcal{H}^{C}_{\rho(j)} \Big) \\
		&=   \bigcup_{\rho \, \in \, P} \, \Big(  \bigcap_{\, a_j \in \mathcal{C}} \mathcal{H}^{C}_{\rho(j)}  \cap \Delta_m \Big) \\
		&= \bigcup_{\rho \, \in \, P} \mathcal{\widetilde{U}}_{\rho} ,
	\end{align*}
	where the equalities above follow by iteratively employing De Morgan's laws. Thus, the set $\Delta_m \setminus \bigcup_{a_k \in \mathcal{C}} \mathcal{U}_k$ is non-empty with probability at least $1-\zeta$, as observed before, and it can be decomposed into $|{P}|$ convex sets $\mathcal{\widetilde{U}}_{\rho}$ with $\rho \in {P}$. As a result, we can employ Algorithm~\ref{alg:sample_point}, since $\mathcal{\widetilde{U}}_{\rho}$ is convex. Consequently, we can select the $\rho \in {P}$ such that $\mathcal{\widetilde{U}}_{\rho}$ has non-empty interior. Since the set ${P}$ has a size of at most $\mathcal{O}(n^n)$, we can enumerate all the elements of $P$ in polynomial time, when the number of follower's actions $n$ is fixed. 
	
	Finally, for every region $\mathcal{\widetilde{U}}_{\rho}$ with non-empty interior, for $\rho \in P$, there exists a subset of its vertices $V$ with dimension $m$, ensuring $\textnormal{co}(V) \subseteq \mathcal{\widetilde{U}}_{\rho}$ and $\textnormal{vol}(co(V)) > 0$. Consequently, by enumerating all possible subsets of the vertices defining the regions $\mathcal{\widetilde{U}}_{\rho}$ with $m$ elements, and observing that the number of such vertices is at most $\binom{m+n^2}{m} \le (m+n^2)^m$, we are required to check a maximum of $\binom{(m+n^2)^m}{m}$ regions defined by a set of these $m$ vertices. This results in a complexity of $O(n^{2m^2})$ when $m$ is constant.
	This enumeration allows us to compute a set $V$ containing $m$ linearly independent vertices of a polytope $\mathcal{\widetilde{U}}_{\rho}$ for some $\rho \in P$. We observe that Algorithm~\ref{alg:sample_point} actually needs just $m$ linearly independent vertices of the polytope $\mathcal{P}$ (see line~\ref{line:sample_li_vertices}) and no other parameter. Thus, we can execute $\texttt{Sample-Int}(\mathcal{\widetilde{U}}_{\rho})$ by assigning $\mathcal{V} \gets V$  at line~\ref{line:sample_li_vertices} Algorithm~\ref{alg:sample_point}, where $V$ is the set of $m$ linearly independent vertices found above.
	
	 As a result, since sampling from the set $\Delta_m \setminus \bigcup_{a_k \in \mathcal{C}} \mathcal{U}_{a_k}$ can be efficiently done when either $n$ or $m$ is constant, and noting that the running time of all other steps in Algorithm~\ref{alg:learning_commitment} can be performed in polynomial time, we deduce that the overall running time of Algorithm~\ref{alg:learning_commitment} is polynomial if either the number of follower's or leader's actions is a fixed parameter. 
\end{proof}

\section{Extension to the Case with Equivalent Follower's Actions }\label{sec:app_extension}


In this section, we show how to extend our results to the general case in which there could equivalent follower's actions, where two actions $a_j, a_k \in A_f$ are \emph{equivalent} if $u_f(a_i,a_j) = u_f(a_i,a_k)$ for all $a_i \in A_\ell$.

Our algorithm can be easily extended to handle such a case. To do this, we introduce the following additional notation. First, we introduce the notion of \emph{leader separating hyperplane}. More formally, given a pair of follower's actions $a_j, a_k \in A_f$ such that $a_j \neq a_k$, we let $\mathcal{H}^{\ell}_{jk} \subseteq \mathbb{R}^m$ be the halfspace in which $a_j$ is (weakly) better than $a_k$ in terms of leader's utility, where:
\[
\mathcal{H}^{\ell}_{j k} \hspace{-0.5mm} \coloneqq \hspace{-0.5mm} \left\{  p \in \mathbb{R}^{m} \hspace{-0.5mm} \mid \hspace{-0.5mm} \sum_{a_i \in A_\ell} \hspace{-0.5mm} p_i \big( u_{\ell}(a_i,a_j) \hspace{-0.5mm} - \hspace{-0.5mm} u_{\ell}(a_i, a_k) \big) \hspace{-0.5mm} \geq \hspace{-0.5mm} 0  \right\} \hspace{-0.5mm} .
\]
Furthermore, we let $H^{\ell}_{j k} \coloneqq \partial \mathcal{H}^{\ell}_{jk}$ the hyperplane defining the boundary of the halfspace $\mathcal{H}^{\ell}_{jk}$, which we call the {leader separating hyperplane} between $a_j$ and $a_k$. In addition, given a follower's action $a_j \in {A}_f$, with an abuse of notation we define $A_f(a_j)$ as the set of follower's actions $a_k \in A_f : a_j \neq a_k$ that are equivalent to $a_j$. Formally:
\[
A_f(a_j) \coloneqq \left\{a_k \in {A}_f  \mid a_k \neq a_j \wedge u_{f}(a_i,a_j) = u_{f}(a_i, a_k) \,\, \forall a_i \in \mathcal{A}_l \right\}.
\]
Then, we can re-define the best-response region of action $a_j \in A_f$ as:
\[
\widetilde{\mathcal{P}}_j \coloneqq \Delta_{m} \cap \Bigg(  \bigcap_{a_k \in A_f: a_k \neq a_j} \mathcal{H}_{jk} \Bigg) \cap \Bigg(  \bigcap_{ a_k \in A_f(a_j)} \mathcal{H}^{\ell}_{jk} \Bigg).
\]
Intuitively, the region $ \widetilde{\mathcal{P}}_j $ corresponds to the region in which $a_j$ is a follower's best response when also accounting for tie-breaking, which is done in favor of the leader.
Consequently, we can apply Algorithm~\ref{alg:learning_commitment} and achieve equivalent results. This is because the regions $\widetilde{\mathcal{P}}_j$ are still polytopes, with the only difference being that the number of vertices defining them is bounded by $\binom{m+2n-1}{m}$, as the number of hyperplanes that define their boundary also takes into account leader separating hyperplanes. Furthermore, by letting $L$ be the bit-complexity of leader's payoffs, it is possible to restate Theorem~\ref{thm:main_thm} as:

\begin{theorem}\label{thm:main_thm_2}
	Given any $\zeta \in (0,1)$, with probability at least $1-\zeta$, Algorithm~\ref{alg:learning_commitment} terminates with $p^\star$ being an optimal strategy to commit to, by using a number of samples of the order of $\widetilde{\mathcal{O}}\left(n^2\left(m^7L\log(\nicefrac{1}{\zeta})+ \binom{m+2n}{m} \right)\right)$, 
\end{theorem}

Theorem~\ref{thm:main_thm_2} shows that we can achieve similar results to those obtained in Theorem~\ref{thm:main_thm}, even in cases where there equivalent follower's actions are allowed, by using Algorithm~\ref{alg:learning_commitment}. Specifically, when the number of leader's actions is fixed, the number of samples required to compute an optimal commitment is of the same order as when there are no equivalent follower's actions. In contrast, when the number of follower's actions is fixed, the required number of samples is $\widetilde{\mathcal{O}}(m^{2n})$, differently from the case with no coinciding follower's actions, where it is equal to $\widetilde{\mathcal{O}}(m^{n})$.

\end{document}